\documentclass[12pt]{amsart}    

\usepackage{amsfonts}
\usepackage{amsmath}
\usepackage{graphicx,color}
\usepackage{amsfonts} \usepackage{amsmath, amsthm}
\usepackage{epsfig} \usepackage{verbatim}
\usepackage{amssymb, amsthm, amsmath, bm, tikz, enumerate, color}%, %mathtools}

\newtheorem{theorem}{Theorem}

\newtheorem{definition}{Definition}

\newtheorem{lemma}{Lemma}

\newtheorem{proposition}{Proposition}

\newcommand{\vol}{\mathop{\mathrm{vol}}\nolimits}

\newcommand{\dist}{\mathop{\mathrm{dist}}\nolimits}

\begin{document}

\title[LTE for certain stochastic models]{Local thermal equilibrium for certain stochastic models of heat transport}
\author{Yao Li}
\address{Yao Li: Courant Institute of Mathematical
Sciences, New York University, New York, NY 10012, USA}
\email{yaoli@cims.nyu.edu}

\author{P\'eter N\'andori}
\address{P\'eter N\'andori: Courant Institute of Mathematical
Sciences, New York University, New York, NY 10012, USA}
\email{nandori@cims.nyu.edu}

\author{Lai-Sang Young}
\address{Lai-Sang Young: Courant Institute of Mathematical
Sciences, New York University, New York, NY 10012, USA}
\email{lsy@cims.nyu.edu}
\thanks{LSY was supported in part by NSF Grant DMS-1363161. }

\maketitle
\begin{abstract} This paper is about nonequilibrium steady states (NESS) of a class of
stochastic models in which particles exchange energy with their ``local environments"
rather than directly with one another. The physical domain of the system can be a
bounded region of $\mathbb R^d$ for any $d \ge 1$. We assume that the temperature 
at the boundary of the domain is prescribed and is nonconstant, so that the system 
is forced out of equilibrium. Our main result is local thermal
equilibrium in the infinite volume limit. In the Hamiltonian context, this would mean
that at any location $x$ in
the domain, local marginal distributions of NESS tend to a probability with
density $\frac{1}{Z} e^{-\beta (x) H}$, permitting one to define the local temperature
at $x$ to be $\beta(x)^{-1}$. We prove also that in the infinite volume limit,
the mean energy profile of NESS satisfies Laplace's equation for the prescribed
boundary condition. Our method of proof is duality: by reversing the sample paths
of particle movements, we convert the problem of studying local marginal energy 
distributions at $x$ to that of joint hitting distributions of  certain 
random walks starting from $x$, and prove that the walks in question become increasingly independent as system size tends to infinity.
\end{abstract}

\bigskip
\section{Introduction}

This paper attempts to address, using highly idealized models, two
of the major challenges in nonequilibrium statistical mechanics:
One is the derivation of the Fourier Law, equivalently the heat equation,
 from microscopic principles.
The other has to do with local thermal equilibrium (LTE) for systems that are 
driven out of, and possibly far from,  equilibrium. Both of these topics are of fundamental 
importance, and have been the focus of much research activity in recent years
(see below), yet no satisfactory general theory has been proposed. In this paper
we study the invariant measures corresponding to nonequilibrium
steady states (NESS) of a specific class of particle systems with stochastic interactions.
The models we consider are simple enough to be amenable to rigorous analysis,
yet not overly specialized, so they may offer insight into more general situations.
A feature that distinguishes our paper from a number of previous works 
is that
the physical space of our models can be $\mathbb R^d$ for any $d \ge 1$.
Model behavior depends on $d$, necessitating different arguments in the proofs, 
but our results are valid for all $d \ge 1$.

We begin with a rough model description; see Sect. 2.1 for
more detail. For any $d \ge 1$, let $\mathcal D \subset \mathbb R^d$ be a bounded domain with smooth boundary, and let $T$
be a prescribed temperature function on $\partial \mathcal D$. We consider
 $\sim L^d$ particles performing independent random walks on 
$\mathcal D_L = \mathbb Z^d \cap L \mathcal D$ where
$L\gg 1$ is a real number and $L \mathcal D$ is the dilation of $\mathcal D$.
These particles do not interact directly with one another, but only via
their ``local environments", symbolized by a collection of random variables each representing the energy at a site in $\mathcal D_L$. More precisely, each particle
carries with it an energy. As it moves about, 
it exchanges energy with each of the sites it visits, and when it reaches the boundary
of $\mathcal D_L$, it abandons the energy it was carrying, replacing it by an energy 
drawn randomly from the ``bath distribution"  at the corresponding point in $\partial \mathcal D$.

These stochastic models are modifications of the 1-D mechanical chains studied in
\cite{EY06} and their 2-D generalizations in \cite{LY10}. In these mechanical models,
energy transport occurs via particle-disk interactions, an idea borrowed from 
\cite{LLM03}. More precisely, there is an array of rotating disks evenly 
spaced in the domain. Particles do not interact with each other directly;
they exchange energy with these disks upon collision.
Our site energies are an abstraction of the kinetic energies of these disks, 
or the ``tank energies" in the stochastic models
in \cite{EY06}. Further simplifications have been introduced 
in the present models to make the analysis feasible.

In our models, when the prescribed temperature function is constant on $\partial
\mathcal D$, i.e. $T \equiv T_0$ for some $T_0 \in \mathbb R^+$,
 it is easy to see that the unique invariant probability distribution is a product
measure: particle numbers are independent and Poissonian, and particle and site
energies are independent  and exponentially distributed 
with mean $T_0$. Let us refer to such a distribution as
``the {\it equilibrium distribution} at temperature $T_0$". For a nonconstant $T$ (all that
we require is that it be a continuous function), the system is forced out of equilibrium.
It is not hard to see that there is still a unique NESS to which all initial states converge.
Our two main results assert that the following hold in the infinite volume limit: 

\smallskip 
(1) mean energy profiles with respect to NESS when scaled back to $\mathcal D$ 

\quad converge  to the unique solution $u(x)$ of 
$$\Delta u=0 \quad \mbox{ on } \mathcal D\ ,
\qquad u|_{\partial \mathcal D} = T\ ;
$$ 

(2) given any $x \in \mathcal D$, for sites located near $xL$, 
marginal distributions of 

\quad the NESS tend to the equilibrium distribution at temperature $u(x)$.

\smallskip \noindent
These and other results are formulated precisely in Sect. 2.2.

Our method of proof is {\it duality}, and the dual used here is similar to that in  
\cite{RY07}, which in turn borrowed its main idea from \cite{KMP82}. We
differ from these earlier works in that the KMP-model is
not a particle system, and our results are proved for all $d\ge 1$. As in \cite{KMP82},
our ``dual process" (duality is with respect to a function,
to be precise) keeps track of movements of certain discrete objects we call ``packets"
in this paper. Reasoning naively,  marginal energy distributions at a site $v \in
\mathcal D_L$ is determined
by what particles bring to this site from the bath, so of interest are the points of origin 
of these energies. The idea is that to identify these points of origin, we can
place some packets at site $v$, to be carried around by particles in a manner 
analogous to the way energy is transported, run the particle trajectories ``backwards", 
and look at the hitting distributions of the packets on the boundary of $\mathcal D_L$. 
As it turns out, packet movements in this process are
effectively independent random walks when their trajectories do not meet. 
From this, one observes that the problem is simpler in $d \ge 3$, where independent
random walkers tend not to meet, and that for $d=1$, where they meet often, 
exchangeability, which was used in \cite{KMP82}, is entirely natural.

\bigskip \noindent
{\bf Related works.} We finish with a discussion of some results that are related to the
various aspects of the present paper; needless to say, this is very far from
a complete list. Some general references on interacting particle
systems are \cite{L85,S91,KL99}. A sample of the articles on NESS are 
\cite{RLL67,KLS84,EPR99,RT02,DLS02,D07,FLM11,KY13,CE14,LY14}.
These and other works treat the existence and uniqueness of NESS, 
correlation decay, fluctuations etc. in stationary states. 
Two informative reviews on the Fourier Law are
\cite{BLRB00} and \cite{LLP03}. Mathematical papers on the Fourier Law proven
under various assumptions include \cite{BO05, BLL04, BK07,R11,DL11,LO12,DN14}.
There are many papers in the physics literature; see
the references in \cite{LLP03} and e.g. \cite{GG08}.

Turning now to LTE, for the physics of this phenomenon, see \cite{GM62}.
Our literature search turned up relatively few rigorous results on LTE: 
An early important example is \cite{KMP82},
see also \cite{RY07} and \cite{OR13};
a systematic study of certain gradient models in \cite{ELS90, ELS91};
and the survey article \cite{CGGR13}. It has been noted that a number of systems 
with integrable dynamics do not have LTE, see e.g.  \cite{DD99,BLY10,RY12}. 
In \cite{EY06} and \cite{LY10}, the authors derived formulas for macroscopic
observations such as energy and particle density profiles for both mechanical and stochastic models 
{\it assuming} LTE, and provided numerical validation for their derivations. 

As for duality, two classical references are \cite{S70} and \cite{L85}.
Two early influential papers in the use of this technique are \cite{KMP82} and 
\cite{PS83}; later works include e.g. \cite{GKR07,BCS12}. 
A recent survey of duality from the probabilistic viewpoint is \cite{JK14}, and
a nice survey on the use of duality in statistical mechanics problems is 
\cite{CGGR13}; this survey contains many other references.

%%%%%%%%%%%%%%%%%%%%%%%%%%%%%%%%%%%%%%%%
%%%%%%%%%%%%%%%%%%%%%%%%%%%%%%%%%%%%%%%

\medskip
\section{Model and Results}
\label{secdef}

\subsection{Model description} \label{subsecdef}

We fix a dimension $d \geq 1$, and let $\mathcal D \subset \mathbb R^d$ be 
either a bounded open rectangle or a bounded, connected open set for which 
$\partial \mathcal D$ is a $\mathcal C^2$ submanifold. A continuous function 
$T : \partial \mathcal D \rightarrow \mathbb R_{+}$ to be thought of as {\it temperature}
is prescribed. 
For $L \gg 1$, the physical domain of our system is  
$$\mathcal D_{L} = L \mathcal D \cap \mathbb Z^d\ ,$$
and each lattice point in $\mathcal D_L$ is referred to as a {\it site}.
The {\it bath} is located at
$$
\mathcal B_L = \partial \mathcal D_L \setminus \mathcal D_L $$
where $\partial \mathcal D_L = \{ v\in \mathbb Z^d: v$ has a neighbor in $\mathcal D_L$
and a neighbor outside of $\mathcal D_L\}$. 
Throughout this article, $|\cdot|$ denotes the cardinality of a finite set. 
For any $v \in \mathbb R^d$,
$\langle v \rangle$ denotes its closest point in $\mathbb Z^d$ and $\| v\|$ is the Euclidean norm.

\medskip
We consider a Markov process $\bm X_t = \bm X_t^{(L)}$ with random variables 
$$
 \underline x = ((\xi_{v})_{v \in \mathcal D_{L}}, (\eta_1, ..., \eta_M), (X_1, ..., X_M)) 
\in \mathbb R_+^{\mathcal D_{L}} \times \mathbb R_+^{M} 
\times \mathcal D_{L} ^M 
$$
where $\xi_{v}$ denotes the energy at site $v \in \mathcal D_L$, 
$M=M(L) \in \mathbb Z^+$ is the number of particles in the system, and $\eta_i$ and 
$X_i$ are the energy and location of particle $i$.
The infinitesimal generator of this process has the form
$$
(Gf)(\underline x) =  (G_1f)(\underline x) + (G_2f)(\underline x),
$$
where $G_1$ and $G_2$ describe respectively interactions within $\mathcal D_L$ and
with the bath: 
Each particle carries an exponential clock which rings at rate $1$ independently
of the clocks of other particles. 
When its clock rings, the particle exchanges energy with the site at which it is located; 
immediately thereafter it jumps to a neighboring site, choosing
the $2d$ nearest neighbors with equal probability. The $G_1$ part of the generator
describes the action when the neighboring site chosen is in
$\mathcal D_L$:
\begin{eqnarray*}
 (G_1f)(\underline x) &=&
\sum_{k=1}^M  \frac{1}{2d} \sum_{w \in \mathcal D_L: |X_k -w| =1} \int_0^1 \mathrm{d}p \\
&& \qquad \Big[ 
f( (\xi'_{v})_{v \in \mathcal D_{L}},
\eta_1, ..., \eta_{k-1}, (1-p)(\xi_{X_k} + \eta_k), \eta_{k+1}, ..., \eta_M,\\
&&  \qquad X_1, ..., X_{k-1}, w, X_{k+1}, ..., X_M)
-f(\underline x)
\Big],
\end{eqnarray*}
where 
\begin{equation}
\label{eq:xihat}
\xi'_v = \left\{ \begin{array}{rl}
 \xi_v &\mbox{ if $v \neq X_k$} \\
  p(\xi_{X_k} + \eta_k) &\mbox{ if $v = X_k$.}
       \end{array} \right.
\end{equation}
If the particle jumps to a site $v \in \mathcal B_L$, then its energy is updated according
to the temperature at $v$, and it is returned immediately to its original site. More precisely,
we extend $T$ to a neighborhood of $\partial \mathcal D$, and define
\begin{eqnarray*}
 (G_2f)(\underline x) &=&
\sum_{k=1}^M  \frac{1}{2d}  \sum_{w \in \mathcal B_L: |X_k -w| =1} 
\int_0^1 \int_0^{\infty} \mathrm{d}p \mathrm{d} \eta' \beta \left(\frac{w}{L}\right) 
e^{- \eta' \beta(\frac{w}{L})}\\
&& \quad \Big[f((\xi'_{v})_{v \in \mathcal D_{L}}, \eta_1, ..., \eta_{k-1}, \eta', \eta_{k+1}, ..., 
\eta_M, X_1, ..., X_M) - f(\underline x) \Big],
\end{eqnarray*}
where $\xi_v'$ is given by (\ref{eq:xihat}) and $\beta(\frac{w}{L})=T(\frac{w}{L})^{-1}$.

This completes the definition of our model.

\bigskip \noindent
{\it Remark on distinguishable vs indistinguishable particles.} As defined,
the particles in $\bm X_t$ are named and distinguishable.
Since our results pertain to infinite-volume limits, it is natural 
to work with models with {\it indistinguishable} particles.
For each $L$, such a model can be obtained from $\bm X_t^{(L)}$ via the following identification:
 for $\underline x, \underline y \in
\mathbb R_+^{\mathcal D_{L}} \times \mathbb R_+^{M} \times \mathcal D_{L} ^M$, 
we let $\underline x \sim \underline y$ if 
\begin{eqnarray*}
\underline x & =  & ((\xi_{v})_{v \in \mathcal D_{L}}, (\eta_1, ..., \eta_M), (X_1, ..., X_M))\\
\mbox{and } \quad 
\underline y & = & ((\xi_{v})_{v \in \mathcal D_{L}}, (\eta_{\sigma(1)}, ..., \eta_{\sigma(M)}), 
(X_{\sigma(1)}, ..., X_{\sigma(M)}))
\end{eqnarray*}
where $\sigma$ is a permutation of the set $ \{1,2,\dots, M\}$.
It is easy to check that the quotient process $\bm X_t^{(L)}/_\sim$ is well defined and
corresponds to $\bm X^{(L)}_t$ with  indistinguishable particles.
As the desired results  for
$\bm X_t^{(L)}/_\sim$ are deduced easily from those for $\bm X_t^{(L)}$,
we will, for the most part, be working with $\bm X_t^{(L)}$.

%%%%%%%%%%%%%%%%%%%%%%%%%%%%%%%%%%
\subsection{Statement of results}
\label{subsec:res}

We begin with a result on the existence and uniqueness of  invariant measure 
in the equilibrium case, i.e., when the prescribed bath temperature $T$ is constant.

\begin{proposition}
\label{prop:equilibrium}
Let $d \in \mathbb Z^+$ and $\mathcal D \subset \mathbb R^d$ be as above,
and let $L$ be such that $\mathbb Z^d$ restricted to $\mathcal D_L$
is a connected graph. If the function $T$ is constant, then with the notation $\beta = 1/T$,
$$ \mu^{(L)}_e = \prod_{v \in \mathcal D_L} \beta e^{- \beta \xi_v} d\xi_v \prod_{j=1}^M \beta e^{- \beta \eta_j} d\eta_j
\prod_{k=1}^M \left( \frac{1}{|\mathcal D_L|} \sum_{v \in \mathcal D_L} \delta_{X_{k} =v} \right)$$
is the unique invariant probability measure of the process $\bm X_t^{(L)}$.
\end{proposition}

We are primarily interested in out-of-equilibrium settings defined by non-constant 
bath temperatures.

\begin{proposition}
\label{prop:invmeasexists}
Let $d, \mathcal D$ and $T$ be as in Sect. 2.1. We assume $L$ is such that $\mathbb Z^d$ restricted to $\mathcal D_L$
is a connected graph. Then the process ${\bm X_t^{(L)}}$ has a unique invariant probability measure $\mu^{(L)}$.
Furthermore, the distribution of ${\bm X_t^{(L)}}$ converges to $\mu^{(L)}$ as $t \rightarrow \infty$
for any initial distribution of ${\bm X_0^{(L)}}$.
\end{proposition}

The proofs of Propositions \ref{prop:equilibrium} and \ref{prop:invmeasexists} are 
straightforward; thus we mention only the ideas and leave details to the reader.
For Propositions \ref{prop:equilibrium}, one can check by a direct computation that 
$\mu^{(L)}_e$ is invariant; uniqueness follows from Doeblin's condition. 
Once Proposition \ref{prop:equilibrium} is established, tightness (at ``$\infty$"
and at ``$0$") for Proposition \ref{prop:invmeasexists} can be proved as follows: 
Given ${\bm X_t^{(L)}}$ defined by a continuous bath temperature function 
$T|_{\partial \mathcal D}$, we consider two equilibrium processes corresponding
to boundary conditions $T^{\rm max} = \sup_{x \in \partial \mathcal D} T(x)$ and
$T^{\rm min} = \inf_{x \in \partial \mathcal D} T(x)$. Let
$\mu^{(L)}_{e, \rm max}$ and $\mu^{(L)}_{e, \rm min}$ be the invariant probabilities
of these two processes respectively. Coupling $\bm X^{(L)}_t$ to the 
process defined by $T^{\rm max}$ in the natural way with 
$\bm X^{(L)}_0 = \mu^{(L)}_{e, \rm max}$, 
we see that the distribution of $\bm X^{(L)}_t$ is stochastically dominated by 
$\mu^{(L)}_{e, \rm max}$. Likewise, coupling to the process defined by $T^{\rm min}$,
we see that the distribution of $\bm X^{(L)}_t$ is stochastically bounded from below
by $\mu^{(L)}_{e, \rm min}$. 

\bigskip
{\it In Theorems 1--4 below, it is assumed that 
$d, \mathcal D, T$ and $L$ are as in Sect. 2.1.}
For given $L$, we let $\mathbb E^{(L)}( \cdot )$ denote expectation with respect to 
the nonequilibrium steady state $\mu^{(L)}$. 
Our first result is on the profile of mean site energies.

\medskip
\begin{theorem} [\bf Mean site energy profiles] \label{thm1}
For each $x \in \mathcal D$,
$$
\lim_{L \to \infty} \mathbb E^{(L)}(\xi_{\langle xL \rangle}) \ = \ u(x)
$$
where $u$ is the unique solution of the equation
\begin{equation}
\label{eq:defu}
 \Delta u = 0\text{ on } \mathcal D\ , \qquad u|_{\partial \mathcal D} = T.
\end{equation}
\end{theorem}

\medskip

Next we proceed to the definition of LTE
for site energies. For $x\in \mathcal D$ and a finite set $S \subset \mathbb Z^d$, 
we denote by $\mu_{x,S}^{(L)}$ the projection of 
$\mu^{(L)}$ to the coordinates $(\xi_{\langle xL \rangle + v})_{v \in S}$ 
and identify $\mu_{x,S}^{(L)}$ with a measure on $\mathbb R^S$
with coordinates $(\zeta_s)_{s \in S}$. Given $M(L) \in \mathbb Z^+$, 
we say the site energies $\xi_v$ of ${\bm X_t}^{(L)}$ approach {\it local thermodynamic equilibrium} (LTE) as $L \to\infty$ if for every $x \in \mathcal D$ and
every finite set $S \subset \mathbb Z^d$, 
$$ \mu_{x,S}^{(L)} \Rightarrow \mu_{x,S} = \prod_{s \in S}\beta(x) e^{-\beta(x) \zeta_s} d \zeta_s
\qquad \mbox{ as } L \to \infty\ ,
$$
where $\Rightarrow$ stands for weak convergence and $\beta(x)=u(x)^{-1}$,
$u$ being the function in Theorem \ref{thm1}.
%{\it Important remark:} 
%\footnote{ The definition of $\mu_{x,S}^{(L)}$, and several other definitions in the sequel 
%only make sense for $L > L_0(\mathcal D, x)$. Since we are 
%always interested in the limits as $L \rightarrow \infty$, we disregard this fact.}

\medskip
\begin{theorem}[\bf LTE for site energies]  \label{thm2}
For any choice of $M(L) \in \mathbb Z^+$, site
 energies $\xi_v$ of ${\bm X_t}^{(L)}$ approach LTE as $L \to \infty$\ .
\end{theorem}

\medskip
In the case where the number of particles tends to a fixed positive
density as $L \to \infty$, i.e. ,
\begin{equation}
\label{eq:densetracers}
 \frac{M(L)}{L^d} \to \alpha \vol (\mathcal D) \qquad \mbox{ as } L \to \infty
\end{equation}
for some constant $\alpha>0$, a more complete notion of LTE should include not only
distributions of site energies but also those of particle energies. In preparation for the
formal statement, we introduce the following notation: 

Let $\mu^{(L)} /_\sim$
be the steady state distribution of the process $\bm X^{(L)}_t/_\sim$ (see the Remark
at the end of Sect. 2.1). Given $x \in \mathcal D$, $S =\{ v_1, ..., v_s \}\subset \mathbb Z^d$, and non-negative integers $K_1, ..., K_s$, we consider the conditional probability
$$
(\mu^{(L)} /_\sim) \ | \ \{\# \mbox{ particles at site } \langle xL \rangle+v_j = K_j\ ,
\ j=1,\dots, s\}
$$
and project this measure to the site and particle energy coordinates on 
$\langle xL \rangle + S$. The resulting measure,  $\nu^{(L)}_{x, S, K_1, \dots, K_s}$, 
can be viewed as a measure on \linebreak
$ \Pi_{j =1}^s (\mathbb R_+ \times (\mathbb R_+^{K_j})/_\sim)$, where
the relation in $(\mathbb R_+^K)/_\sim$ is defined by $x \sim y$ if
$x=(x_1,...,x_K)$, $y=(x_{\sigma(1)}, ...,x_{\sigma(K)} )$, and $\sigma$ is
a permutation of the set $\{1, \dots, K\}$.
%
%we used $\sim$, with a slight abuse of notation, for the relation 
%$x \sim y$, $x,y \in \mathbb R^n$ if there is a permutation $\sigma$
%such that $x=(x_1,...,x_n)$, $y=(x_{\sigma(1)}, ...,x_{\sigma(n)} )$.
%\marginpar{I added a sentence about $\sim$.}
We are concerned with the limit of  $\nu^{(L)}_{x, S, K_1, \dots, K_s}$ as $L \to \infty$.

\medskip
\begin{theorem}[\bf LTE for systems with positive density of particles]
\label{thm:dense} Suppose $M(L)/L^d \to \alpha \vol (\mathcal D) $ for some
$\alpha>0$. Then the system approaches LTE as $L \to \infty$ in the sense of both 
site and particle energy distributions. That is to say,
for every $x \in \mathcal D$, $S =\{ v_1, ..., v_s \}\subset \mathbb Z^d$, 
and non-negative integers $K_1, ..., K_s$, we have the following limiting distributions
as $L \to \infty$:
\begin{itemize}
\item[(1)] Let $k^{(L)}_{v_j}$ be the number  of particles at site $\langle xL\rangle + v_j$, 
$j=1,2, \dots, s$, seen as a random variable with respect to $\mu^{(L)}$.
Then $k^{(L)}_{v_1}, \dots {k^{(L)}_{v_s}}$ tend in distribution to independent Poisson 
random variables with mean $\alpha$.
\item[(2)] The measures $\nu^{(L)}_{x, S, K_1, \dots, K_s}$ 
converge weakly to 
\begin{equation}
\Pi_{j=1}^s \ \left(\mathcal E_x \ \times (\mathcal E_x^{K_j})/_\sim \right) \qquad
\mbox{ as } L \to \infty\ ,
\end{equation}
where $\mathcal E_x$ is the exponential distribution on $\mathbb R_+$ with parameter
$\beta = u(x)^{-1}$, $\mathcal E_x^K$ is the product of $K$ copies of $\mathcal E_x$, and
$\sim$ is the usual identification in $(\mathbb R_+^K)/_\sim$.
\end{itemize}
\end{theorem}

\medskip
The notion of LTE considered so far describes marginal energy distributions in regions of 
 microscopic sizes. These results can be extended to a version of LTE at {\it mesoscopic} scales:

\medskip
\begin{theorem}[\bf LTE at mesoscopic scales]
\label{thm3}
For any $\vartheta \in (0,1)$, the results in Theorems 2 and 3 remain valid
if we replace 

\quad  `` $(\langle xL\rangle + v)_{v \in S}$ where $S \subset \mathbb Z^d$
is any finite set " 

\noindent
by \quad  \ `` $({\langle xL + L^{\vartheta}v \rangle })_{v \in S} $
where  $S \subset \mathbb R^d$  is any finite set  "  .
\end{theorem}

\medskip
%%%%%%%%%%%%%%%%%%%%%%%%%%%%%%%%%%%%%%%
%%%%%%%%%%%%%%%%%%%%%%%%%%%%%%%%%%%%%
\section{Preliminaries from probability theory}
\label{sec:prob}

We collect in this section some facts from probability theory that will be used.
All the results cited are known, possibly with the
exception of Proposition \ref{lem12}.

\subsection{Random walks}
\label{subsect:rw}

Here, we formulate some basic lemmas about random walks. We will use the terminology 
{\it random walk} for any Markov chain of the form
$$S_n = \sum_{k=0}^n \xi_k,$$
where the $\xi_k$'s are independent, identically distributed random variables with values in $\mathbb Z^d$ ($d$ is the dimension
of the random walk). The special case where $\xi_1$ is 
supported on the origin and its $2d$ nearest neighbors is called 
a {\it nearest neighbor random walk}, and  the case where 
$\xi_1$ is uniformly distributed on the $2d$ nearest neighbors of the origin is called
a {\it simple symmetric random walk} (SSRW).  The following statement is arguably the most important property of random walks with finite variance.

\begin{lemma}[Invariance principle]
\label{lemma:invpr}
 Consider a $d$-dimensional random walk 
$S_n = \sum_{k=0}^n \xi_k$ where $\xi_1$ has zero expectation and 
finite covariance matrix $\Sigma$, and let $W_n =(W_n(t))_{t \in [0,1]}$
be the random process 
defined by $W_n(\frac{k}{n}) =\frac{S_k}{ \sqrt k}$ and linear interpolations 
between $\frac{k}{n}$ and $\frac{k+1}{n}$. Then as $n \rightarrow \infty$,
$W_n$ converges weakly to the $d$-dimensional 
Brownian motion on $[0,1]$ with covariance matrix $\Sigma$. 
\end{lemma}

\begin{lemma}[Harmonic measure]  \label{lemma:harmonic}
 Let $x \in \mathcal D \subset \mathbb R^d$ 
where $\mathcal D$ is as above, and let $\tau$ be the first hitting time of $\partial 
\mathcal D$ for a Brownian motion $B^x_t$ starting from $x$. Then
$$
\mathbb E (T(B^x_\tau)) = u(x)
$$
where $u$ is given by (\ref{eq:defu}).
\end{lemma}

The next result is a combination of the last two, together with a small perturbation
in starting location. A version of this result likely exists in the literature, but we are 
unable to locate a reference. Since our proofs rely heavily on 
Proposition \ref{lem12}, we have included its proof in the Appendix.

\begin{proposition} \label{lem12}
Let $x \in \mathcal D \subset \mathbb R^d$ be as above, and let
$\varepsilon>0$ be given. Then there exist
$\delta>0$ and $L_0$ such that the following holds true for all $L>L_0$: 
 Let $S_n$ be a SSRW on $\mathbb Z^d$ with
$\| S_0 - xL\| < \delta L$, and let $\tau$ be the smallest
$n$ such that $S_n \in \mathcal B_L$. Then 
\begin{equation} \label{approx}
 \left| \mathbb E \left(T \left( \frac{S_{\tau}}{L} \right) \right) -u(x)\right| < \varepsilon.
 \end{equation}
\end{proposition}

\smallskip

We will also use the following estimate on moderate deviations. It is a consequence of e.g. Theorem 1 in \cite{P75}, Section VIII.2:

\begin{lemma}
\label{lemma2}
 Let $\xi_n$ be bounded i.i.d real-valued random variables with variance $\sigma^2>0$. 
 There is a constant $c_1$ such that for any $n \ge 1$ and
for any $R$ with $1 < R < n^{1/6}$,
$$ \mathbb P \left( \big| \sum_{k=1}^n \xi_k - n \mathbb E(\xi) \big|
  > R \sqrt n \right) < c_1 e^{- \frac{R^2}{2 \sigma^2}}.$$
\end{lemma}

The reason for the present review is that instead of
studying $\bm X^{(L)}_t$ directly, we will transform the problem into one involving
certain stochastic processes in which a finite number of walkers perform SSRW on
$\mathbb Z^d$. These walks are independent when the walkers are at distinct sites,
but when they meet, there is a tendency for them to stick together for some random time.
We will need to show that in terms of hitting distributions such as those in Lemma 
\ref{lemma:harmonic}, the situation in the $L \to \infty$ limit is as though the walks were 
independent. This is clearly related to the question of how often two walkers 
meet, a property well known to be dimension dependent. 
\begin{enumerate}
 \item $d=1$. The local time of a SSRW in dimension $1$ at the origin up to time $n$ 
 is $\sim \sqrt n$. More precisely, the local time up to $n$, rescaled by $\sqrt n$
converges weakly to the absolute value of the standard normal distribution (see \cite{CH49}).
Since the typical time needed to leave the interval 
 $[-L,L]$ is $O(L^2)$ by the invariance principle, two independent walkers 
will meet $\sim L$ times before leaving the interval $[-L,L]$. 
%This relatively large
%number of encounters makes it difficult to 
%compare the stochastic processes above to independent walkers.
% 
% This is 'big' in the sense that during the 
%typical time needed to leave the interval $[-L,L]$ (which is of order $L^2$ by the invariance principle) the local time is $L$, hence is comparable to the 
%length of the interval. Thus two independent walkers meet roughly $L$ times before leaving the interval $[-L,L]$. This makes hard to compare our stochastic
%process to independent walkers. On the other hand, dimension $1$ is conceptually the simplest and other tools (like generating functions or martingales) are
%also easy to use. That is why in $d=1$ we will use martingales instead of directly comparing to independent SSRW.
 \item $d=2$. The local time of  a  SSRW in dimension $2$ at the origin up to time $n$ 
 is $\sim \log n$, so that two independent walkers meet $\sim \log L$ times before 
 leaving the domain $\mathcal D_L$. 
 
 \item $d>2$.  SSRWs  in dimensions $d>2$ are transient, meaning two independent 
 walkers only meet finitely many times.
\end{enumerate}
These observations have prompted us to proceed as follows: We will first treat the
$d=2$ case, by comparing the process in question 
to independent walkers. Once that is done, we will 
observe that a simplified version of the argument gives immediately results for $d>2$. 
Dimension one is treated differently: The large number of encounters 
makes it difficult to compare the stochastic processes above to independent 
walkers. Instead, we make use of the meetings of the walkers to show 
that their identities can be switched; see Sect. 3.3.

We collect below two other dimension-dependent facts 
about random walks that will be used in the sequel.

\begin{lemma}[\bf $d=2$] \label{lemma:mesoscop}
Let $\mathcal D \subset \mathbb R^2$ be above. We fix $\beta \in (0,1)$,
$C<\infty$ 
and $\varepsilon >0$, and let $W_n$ be a SSRW on $\mathbb Z^d$. Then 
\begin{equation}
\label{eq:annuli}
 P \left( \min_n \{ \| W_n \| \gtrsim  L \} < \min_n \{  W_n = 0 \}\ | \ 
\| W_0\| \in [ L^{\beta} -C, L^{\beta}]  \right)  > \beta -\varepsilon
\end{equation}
for all $L$ large enough.
\end{lemma}

%\begin{proof} 
It is a well known fact in the probabilistic literature that 
the left hand side of (\ref{eq:annuli}) converges to $\beta$ as $L \rightarrow \infty$ 
(see for instance Proposition 1.6.7 in \cite{L91}). This implies Lemma \ref{lemma:mesoscop}.

For completeness, we provide a heuristic justification for this cited result:
Note that $W_n$ converges to a Brownian motion after rescaling. 
Since the logarithm of a planar Brownian motion is a martingale, it is not difficult to
deduce the statement that if $B_t$ is a planar Brownian motion with $B_0 \neq 0$, 
then the probability that 
$\|B_t\|$ reaches $2\|B_0\|$ before reaching $\|B_0\|/2$ is $1/2$. Consequently,
\begin{equation}
\label{eq:half}
 P \left( \text{ $\|W_n\|$ reaches $2L^{\beta}$ before reaching $L^{\beta}/2$} \right) \sim 1/2.
\end{equation}
For $k \in \mathbb Z$, let
$$\mathcal C_{2^k} = \{ v \in \mathbb Z^2 : \| v\| <2^k\}\ , $$
and define the random variables $t_j, k_j$ in the following way:
$$ t_1= \min_{n > 0} \{ \exists k: W_n \in \partial \mathcal C_{2^k}\} $$
and $k_1$ is such that $W_{t_1} \in \partial \mathcal C_{2^{k_1}}$. Inductively, we define
$$ t_j = \min_{n > t_{j-1}} \{ W_n \in \partial \mathcal C_{2^{k_{j-1}-1}}\text{ or }W_n \in \partial \mathcal C_{2^{k_{j-1}+1}}\} $$
and $k_j$ such that $W_{t_j} \in \partial \mathcal C_{2^{k_j}}$.
Similarly to (\ref{eq:half}), we see that $k_j \approx \log_2 \| W_{t_j}\|$ is approximately a one dimensional SSRW for $L$ large
with starting position $\beta \log_2 L$.
A simple computation (often referred to as gambler's ruin) gives that 
$$ P \left( \min_n \{ \log_2 \| W_{t_n}\| = \log_2 L \} < \min_n \{ \log_2 \| W_{t_n}\| = 0 \} \right) \xrightarrow{L\rightarrow \infty} \beta\ .$$
%\end{proof}

The following is a well known property of high dimensional random walks.

\begin{lemma}[$d>2$]
 \label{lemma:transience}
In dimensions $d>2$, any non-degenerate random walk is transient,
i.e.,  $\|S_n\| \rightarrow \infty$ as $n \rightarrow \infty$ with probability $1$.
\end{lemma}

%%%%%%%%%%%%%%%%%%%%%%%%%%%%%%%%%%
\subsection{Moments of exponential random variables}
\label{sec:moments}

Let $\lambda \in \mathbb R$ and  $s \in \mathbb Z^+$. 
The moments of $s$ independent exponential random variables $X_1, \dots, X_s$
with parameter $\lambda$ are given by
\begin{equation} \label{moments}
m(n_1, \dots, n_s) = \mathbb E (X_1^{n_1} \dots X_s^{n_s}) = \prod_{i=1}^s \frac{n_i !}{\lambda^{n_i}}\ .
\end{equation}
Conversely, if $X_1, \dots, X_s$ are such that their joint moments are given
by (\ref{moments}) for all $(n_1, \dots, n_s) \in \{0,1,2, \dots\}^s$, then they are
independent exponential random variables with parameter $\lambda$. See for instance \cite{Z83}.

%{\bf Add this fact about convergence of a sequence of measures
%with the right moments -- the one used in the proof of Theorem 2 (assuming
%Prop 4)?}

%%%%%%%%%%%%%%%%%%%%%%%%%%%%%%%%%%%

\subsection{Exchangeability}

We also review very briefly the notion of exchangeable random variables, which
will be used in the sequel.
\begin{definition}
 The infinite sequence of random variables $X_1, X_2, ...$ is called {\rm exchangeable}
  if for any finite $N$ and any permutation $\sigma$ of $\{1,2,\dots, N\}$, the random 
vectors $(X_1, ..., X_N)$ and $(X_{\sigma (1)}, ..., X_{\sigma (N)})$ have the same distributions.
\end{definition}

The following result is well known; it was proved by de Finetti in \cite{D31}:

\begin{theorem} [\bf de Finetti's Theorem]
\label{thm:definetti}
If $X_1, X_2, ...$ is a sequence of 
$\{0,1\}$-valued exchangeable random variables, 
then there exists a distribution function $F$ on $[0,1]$
such that for all $n$ and all $x_i \in \{0,1\}$,
$$ P(X_1=x_1, ..., X_n=x_n) = \int_0^1 \theta^{\sum_{i=1}^n x_i} (1- \theta)^{n-\sum_{i=1}^n x_i} \mathrm{d}F(\theta).$$ 
\end{theorem}

%%%%%%%%%%%%%%%%%%%%%%%%%%%%%%%%%%%%%%%%%%
%%%%%%%%%%%%%%%%%%%%%%%%%%%%%%%%%%%%%%%

\medskip
\section{Duality}
\label{sec:duality}

In this section we introduce another process ${\bm Y_t}$ and a function $F$, and show
that ${\bm X_t}$ and ${\bm Y_t}$ are dual with respect to the function $F$ in a sense
to be made precise. We explain also how to leverage duality to prove some of 
the asserted results in Sect. 2.2.

\subsection{Motivation and definition of a ``dual process"}

The idea is as follows:
For the process $\bm X_t=\bm X^{(L)}_t$, consider the marginal energy 
distributions of $\mu^{(L)}$ at $v \in \mathcal D_L$.
These distributions reflect what the particles  bring to 
site $v$ from the bath, and that in turn is reflected in which parts of 
$\mathcal B_L$ a particle visits prior to its arrival 
at site $v$ --- though in reality things are a bit more complicated: particles interact 
with all the sites they pass through on their way to $v$. Nevertheless, accepting
this simplified picture for the moment, we can {\it reverse the 
trajectories of the particles}, and think of them as carrying certain ``packets" from
site $v$ to the bath. The locations in $\mathcal B_L$ at which these packets are deposited
will then tell us from which parts of $\mathcal B_L$ energies were drawn
in the original process, thereby revealing the composition of the steady state distribution
$\mu^{(L)}$ of ${\bm X_t}$ at $v$. An advantage of studying the reverse process is that it is
reminiscent of hitting distributions of Brownian motion, or 
harmonic measures on $\partial D$. These are the ideas behind the duality
formulation discussed in this section.

\bigskip
For $\mathcal D \subset \mathbb R^d$ and $L$ as in Sect. 2.1, we now introduce
a Markov process ${\bm Y^{(L)}_t}$ designed to carry packets from the
sites in $\mathcal D_L$
to $\mathcal B_L$. This process also involves $M=M(L)$ particles. Let $\mathbb N$ 
denote the set of non-negative integers. The variables
in ${\bm Y^{(L)}_t}$ are
$$
\underline n = ((n_{v})_{v \in \mathcal D_L}, (\hat n_{v})_{v \in \mathcal B_L}, \tilde n_1, ..., \tilde n_M, Y_1, ..., Y_M) \in \mathbb N^{\mathcal D_L} \times \mathbb N^{ \mathcal B_L} \times \mathbb N^{M} \times \mathcal D_L^M\ .
$$
Here, $n_{v}$ is to be interpreted as the number of packets at site $v \in \mathcal D_L$, 
$\hat n_{v}$ the number of packets at $v \in \mathcal B_L$, $\tilde n_{j}$ the number of 
packets carried by particle $j$, and $Y_j$ the position of particle $j$.
(We distinguish between packets that have been dropped off at a site
 and packets that are carried by particles: 
$n_v$ counts packets that have been dropped off at site $v$, and does not
include packets carried by particle $j$ even when $Y_j=v$.)
The generator of the process ${\bm Y_t}$ is given by 
$$ (Af)(\underline n) = (A_1f)(\underline n) + (A_2f)(\underline n),$$
where $A_1$ corresponds to movements inside $\mathcal D_L$ and $A_2$ is the part describing the interaction with the boundary. Formally, we have
$$
 (A_1f)(\underline n) \ = \ 
\sum_{k=1}^M  \frac{1}{2d} \sum_{z \in \mathcal D_L: |Y_k -z| =1}
%\sum_{k=1}^M \sum_{i=+1,-1} \bm 1_{ \{Y_k+i \in \{  -L, ..., L\} \} } \frac12 (1+ \bm 1_{ \{ Y_k \in \{  -L-1, L+1\} \} })
\frac{1}{\tilde n_k + n_{z} +1}
 \sum_{q=0}^{\tilde n_k + n_{z} } [f(\cdots) - f(\underline n) ]
$$
where the quantity  inside the parenthesis is
$$
(n'_{v})_{v \in \mathcal D_L}, (\hat n_{v})_{v \in \mathcal B_L},
\tilde n_1, ..., \tilde n_{k-1}, (\tilde n_k + n_{z} -q) , \tilde n_{k+1}, ..., \tilde n_{M},
Y_1, ..., Y_{k-1}, z, Y_{k+1}, ..., Y_M
$$
with
\begin{equation}
\label{eq:n'}
n'_v = \left\{ \begin{array}{rl}
 n_v &\mbox{ if $v \neq z$} \\
  q &\mbox{ if $v = z$.}
       \end{array} \right.
\end{equation}
Note the difference between (\ref{eq:xihat}) and (\ref{eq:n'}): In ${\bm X_t}$,
when a particle's clock rings, it pools its energy together with that at the site at
which it is located, and carries a random fraction of it as it jumps to a new site.
In ${\bm Y_t}$, the order is reversed: when a particle's clock rings, it jumps
to a new site, then pools together the packets it is carrying with those at the new
site, and takes a random fraction, assuming this new site $z$ is inside $\mathcal D_L$.  
The second half of the generator treats the case where $z \in \mathcal B_L$:
$$
(A_2f)(\underline n) =
\sum_{k=1}^M  \frac{1}{2d}  \sum_{z \in \mathcal B_L: |Y_k -z| =1}
\frac{1}{n_k +1} \sum_{q=0}^{n_k } [f(\cdots) - f(\underline n)]
$$
where the quantity inside the parenthesis is
$$
(n'_{v})_{v \in \mathcal D_L}, (\hat n'_{v})_{v \in \mathcal B_L},
\tilde n_1, ..., \tilde n_{k-1}, n_{k} -q , \tilde n_{k+1}, ..., \tilde n_{M}, Y_1, ..., Y_M\ ,
$$
with
\begin{equation*}
\hat n'_v = \left\{ \begin{array}{rl}
 \hat n_v &\mbox{ if $v \neq z$} \\
  \hat n_v+\tilde n_k &\mbox{ if $v = z$}
       \end{array} \right. \quad \text{ and  }\quad
n'_v = \left\{ \begin{array}{rl}
 n_v &\mbox{ if $v \neq Y_k$} \\
  q &\mbox{ if $v = Y_k$.}
       \end{array} \right.
\end{equation*}
That is to say, if particle $k$ jumps from site $v \in \mathcal D_L$ to site 
$z \in \mathcal B_L$, then the following occurs instantaneously: it drops off all 
of the packets it is carrying at site $z$, returns to site $v$, and takes a random fraction 
of the packets located at site $v$. Once a packet is dropped off in $\mathcal B_L$, it will
remain there permanently,
so that as time tends to infinity there will be no packets left in $\mathcal D_L$.

\medskip
This completes the definition of the process ${\bm Y^{(L)}_t}$.
%%%%%%%%%%%%%%%%%%%%%%%%%%%%%%%%%%%%%%

\subsection{Proof of pathwise duality}

The function with respect to which duality will be proved is 
\[ F(\underline n, \underline x) = 
\prod_{v \in \mathcal D_L} \frac{\xi_v^{n_v}}{n_v !}
\prod_{j=1}^{M} \frac{\eta_{j}^{\tilde n_{j}} }{\tilde n_{ j} !}
\prod_{v \in \mathcal B_L} \left[ T\left( \frac{v}{L} \right) \right]^{\hat n_v}
\]
where $\underline x,  \xi_v$ and $\eta_j$ are as in the definition of ${\bm X_t}$
and the rest are from the definition of $\bm Y_t$.
Notice that $F$ does not depend on the positions of the particles in either 
${\bm X_t}$ or ${\bm Y_t}$. For reasons to become clear, it is convenient
to write
$$
\underline x = (\underline{\check{x}}, \bar X) \quad \mbox{ and } \quad
\underline n = (\underline{\check{n}}, \bar Y)
$$
where $\underline{\check{x}} \in \check{\mathcal X} := \mathbb R_+^{\mathcal D_{L}} \times \mathbb R_+^{M}$ denotes the energy coordinates and $\bar X = (X_1, \dots, X_M)$
the positions of the particles. Writing 
$\bm X_t = (\underline{\check{x}}_t, \bar X_t)$, we observe that $\bar X_t$ consists of
$M$ independent continuous-time random walks on $\mathcal D_L$, 
with ``reflection" at $\partial \mathcal D_L$
as defined earlier. Likewise, $\underline{\check{n}} \in 
\check{\mathcal Y} :=
\mathbb N^{\mathcal D_L} \times \mathbb N^{ \mathcal B_L} \times \mathbb N^{M}$
gives the number of packets at the various sites or carried by particles, while
$\bar Y = (Y_1, \dots, Y_M)$ denotes the positions of the particles; and if 
$\bm Y_t = (\underline{\check{n}}_t, \bar Y_t)$, then $\bar Y_t$
has the same description as $\bar X_t$. Notice that $F$ is really a function of 
$(\underline{\check{x}}, \underline{\check{n}})$.

We now formulate a version of {\it pathwise duality}, counting sample paths of
$\bar X_t$ (or $\bar Y_t$) in the following way: We say a  ``move" in $\bar X_t$ 
occurs when either (a) a particle jumps from 
$z \in \mathcal D_L$ to $w \in \mathcal D_L$ or (b) it jumps -- all in one instant
-- from  
$z \in \mathcal D_L$ to $w \in \mathcal B_L$ and back, and we regard different
$w \in \mathcal B_L$ as corresponding to distinct sample paths.
Now fix a time interval $(0, \tau)$, and let $t_1< t_2 < \dots < t_n$ be the times in
$(0,\tau)$ at which $\bar X_t$ moves. To avoid discussing $t_i^\pm$ (i.e. just before 
or after $t_i$), we let $s_0=0, s_n=\tau$, fix arbitrarily
$s_i \in (t_i, t_{i+1})$ for  $i=1,\dots, n-1$, and agree to abbreviate this sample
path as $\bm \sigma = (\sigma_0, \sigma_1, \dots, \sigma_n)$ where 
$\sigma_i = \bar X_{s_i}$, with the understanding that in case (b),
both the particle and the bath location involved are specified.
We also use the notation $\bm \sigma^{-1} = (\sigma_n, \dots, \sigma_0)$ to
denote the sample path corresponding to $\bm \sigma$ parameterized backwards
in time, 
i.e., for $\bm \sigma^{-1}$, 
moves are made at times $\hat t_1 < \dots < \hat t_n$ where $\hat t_i = \tau - t_{n+1-i}$;
and at times $\hat t_1, \hat t_2, \dots$, the moves of $\bm \sigma^{-1}$ are the reverse of 
those in $\bm \sigma$ at times $t_n, t_{n-1}, \dots$.

Our one-step duality lemma reads as follows:

\medskip
\begin{lemma}
\label{lemmaduality}
 For any fixed $\underline{\check{x}} \in \check{\mathcal X}$,  $\underline{\check{n}} \in \check{\mathcal Y}$, and any sample path 
 $\bm \sigma = (\sigma_0, \sigma_1)$ on $[0,\tau]$, we have
\[ \mathbb E( F(\underline{\check{n}}, \underline{\check{x}}_\tau) | 
\underline{\check{x}}_0 = \underline{\check{x}}, \bm \sigma)
=
\mathbb E( F(\underline{\check{n}}_\tau, \underline{\check{x}}) | \underline{\check{n}}_0 = \underline{\check{n}}, \bm \sigma^{-1})\ . \]
\end{lemma}

\begin{proof} The sample path $\bm \sigma = (\sigma_0, \sigma_1)$ describes
exactly one move on the time interval $(0,\tau)$.
We consider separately the two cases corresponding to the two terms in the
generator $G$ of $\bm X_t$ (see Sect. 2.1).

\medskip
\noindent {\bf Case 1.} Particle $k$ jumps from site 
$z \in \mathcal D_L$ to site 
$w \in \mathcal D_L$. The term corresponding to $k$ and $w$ in $G_1F$ 
can be written as $I \cdot II$, where
$$ I =\frac{1}{2d} \prod_{v \neq z} \frac{\xi_v^{n_v}}{n_v !}
\prod_{j \neq k} \frac{\eta_{j}^{\tilde n_{j}} }{\tilde n_{ j} !}
\prod_{v \in \mathcal B_L} \left[ T\left( \frac{v}{L} \right) \right]^{\hat n_v}.
 $$
and
$$
 II = \frac{1}{n_{z}!} \frac{1}{\tilde n_k !}
\left( \int_0^1 \left( \xi_{z} + \eta_k \right)^{n_{z} + \tilde n_k} p^{n_{z}}
(1-p)^{\tilde n_k} \mathrm{d}p - \xi_{z}^{n_{z}} \eta_k^{\tilde n_k} \right).
$$
From this we deduce that
\begin{eqnarray*}
 II +  \frac{1}{n_{z}!} \frac{1}{\tilde n_k !} \xi_{z}^{n_{z}} \eta_k^{\tilde n_k}  &=& \frac{1}{n_{z}!} \frac{1}{\tilde n_k !}
\left( \int_0^1 \left( \xi_{z} + \eta_k \right)^{n_{z} + \tilde n_k} p^{n_{z}}
(1-p)^{\tilde n_k} \mathrm{d}p \right)\\
&=& \frac{1}{n_{z}!} \frac{1}{\tilde n_k !} \left( \xi_{z} + \eta_k \right)^{n_{z} + \tilde n_k}
 \int_0^1 p^{n_{z}}(1-p)^{\tilde n_k} \mathrm{d}p \\
&=& \frac{1}{n_{z}!} \frac{1}{\tilde n_k !}
\left[ \sum_{m=0}^{n_z+\tilde n_k} \frac{(n_z+\tilde n_k)!}{m!(n_z+\tilde n_k-m)!}
\xi_z^m \eta_k^{n_z+\tilde n_k-m} \right]
\frac{n_{z}! \tilde n_k !}{(n_{z} + \tilde n_k +1)!} \\
&=& \frac{1}{n_{z} + \tilde n_k +1} \left[ \sum_{m=0}^{n_z+\tilde n_k} \frac{\xi_z^m}{m!} \frac{\eta_k^{n_z+\tilde n_k-m}}{(n_z+\tilde n_k-m)!}
\right]
\end{eqnarray*}

Thus $I \cdot II$ is the term corresponding to $A_1F$ in the generator of $\bm Y_t$
with indices $k$ and $z$.

\medskip
\noindent {\bf Case 2.} Particle $k$ jumps from $z \in \mathcal D_L$ to  
$w \in \mathcal B_L$ and back. The term corresponding to $k$ in $G_2F$ and $w \in
\mathcal B_L$ can again be written as $I \cdot II$, where $I$ is as above and 
$$ II = \frac{1}{n_z!} \frac{1}{\tilde n_k !} \left[ \int_0^1 \mathrm{d}p  [p\left( \xi_{z} + \eta_k \right)]^{n_{z}}
\int_0^{\infty} \mathrm{d}\theta \theta^{\tilde n_k} \beta \left( \frac{w}{L}\right) e^{- \theta 
\beta \left( \frac{w}{L}\right)}
- \xi_{z}^{n_{z}} \eta_k^{\tilde n_k} \right]\ .$$
From this we obtain
\begin{eqnarray*}
&& II +  \frac{1}{n_{z}!} \frac{1}{\tilde n_k !} \xi_{z}^{n_{z}} \eta_k^{\tilde n_k} \\
 &=& 
\left( \frac{1}{n_z!}\int_0^1 \mathrm{d}p  [p\left( \xi_{z} + \eta_k \right)]^{n_{z}} \right) \cdot
\left(  \frac{1}{\tilde n_k !} \int_0^{\infty} \mathrm{d}\theta \theta^{\tilde n_k} \beta \left( \frac{w}{L}\right) e^{- \theta \beta \left( \frac{w}{L}\right)} \right)\\
&=& \frac{1}{n_{z} +1} \left[ \sum_{m=0}^{n_z} \frac{\xi_z^m}{m!} \frac{\eta_k^{n_z-m}}{(n_z-m)!} \cdot
\left[ T\left( \frac{w}{L}\right)\right]^{\tilde n_k}
\right].
\end{eqnarray*}
In the last equality, we used the simplified version of the computation of $II$ in Case 1 (by letting $\tilde n_k =0$) and computed an elementary integral.
We conclude that $I \cdot II$ is the term corresponding to $A_2F$ for the indices $k$ and $z$.
\end{proof}

Next we extend Lemma \ref{lemmaduality} to sample paths involving arbitrary
numbers of moves.

\begin{lemma}
\label{lemmaduality2}
 For any fixed $\underline{\check{x}} \in \check{\mathcal X}$, $\underline{\check{n}} \in \check{\mathcal Y}$, and any sample path $\bm \sigma = 
 (\sigma_0, \sigma_1, ..., \sigma_m)$ on $[0, \tau]$, we have
\[ \mathbb E( F(\underline{\check{n}}, \underline{\check{x}}_\tau) | 
\underline{\check{x}}_0 = \underline{\check{x}}, \bm \sigma)
=
\mathbb E( F(\underline{\check{n}}_\tau, \underline{\check{x}}) | \underline{\check{n}}_0 = \underline{\check{n}}, \bm \sigma^{-1})\ . \]
\end{lemma}

\begin{proof}
We prove by induction on $m$, the number of moves. The case $m=1$ is Lemma \ref{lemmaduality}. Assume that we have proved the statement for $\leq m-1$ moves. 
Letting $s_1$ be as defined above, we have the following:
\begin{eqnarray*}
&& \mathbb E( F(\underline{\check{n}}, \underline{\check{x}}_\tau) | 
\underline{\check{x}}_0 = \underline{\check{x}}, (\sigma_0, \sigma_1, \dots, \sigma_m))\\
&=& \int \mathbb E(F(\underline{\check{n}},\underline{\check{x}}_\tau) | \underline{\check{x}}_{s_1} = \underline{\check{x}}', (\sigma_1, ..., \sigma_m))
\cdot P(\underline{\check{x}}_{s_1} = \underline{\check{x}}' |  \underline{\check{x}}_0 = \underline{\check{x}}, (\sigma_0, \sigma_1)) \mathrm{d}\underline{\check{x}}'\\
&=& \int \mathbb E(F(\underline{\check{n}}_{\tau-s_1},\underline{\check{x}}') | \underline{\check{n}}_0 = \underline{\check{n}}, (\sigma_m, ..., \sigma_1))
\cdot P(\underline{\check{x}}_{s_1} = \underline{\check{x}}' |  \underline{\check{x}}_0 = \underline{\check{x}}, (\sigma_0, \sigma_1)) \mathrm{d}\underline{\check{x}}'\\
&=& \int \int F(\underline{\check{ n}}',\underline{\check{x}}') \cdot
P( \underline{\check{n}}_{\tau-s_1} = \underline{\check{ n}}' | \underline{\check{n}}_0 = \underline{\check{n}}, (\sigma_m, ..., \sigma_1))
\cdot P(\underline{\check{x}}_{s_1} = \underline{\check{x}}' |
\underline{\check{x}}_0 = \underline{\check{x}}, (\sigma_0, \sigma_1))
\mathrm{d}\underline{\check{x}}' \mathrm{d}
  \underline{\check{ n}}' \\
&=& \int \mathbb E( F(\underline{\check{ n}}',\underline{\check{x}}_{s_1}) | \underline{\check{x}}_0 
= \underline{\check{x}}, (\sigma_0, \sigma_1) ) \cdot
P( \underline{\check{n}}_{\tau-s_1} = \underline{\check{ n}}' | \underline{\check{n}}_0 = \underline{\check{n}}, (\sigma_m, ..., \sigma_1))
\mathrm{d}\underline{\check{n}}'\\
&=& \int \mathbb E( F(\underline{\check{ n}}_\tau, \underline{\check{ x}} ) | 
\underline{\check{n}}_{\tau-s_1} = \underline{\check{ n}}', \sigma_1, \sigma_0 ) \cdot
P( \underline{\check{n}}_{\tau-s_1} = \underline{\check{ n}}' | \underline{\check{n}}_0 = \underline{\check{n}}, (\sigma_m, ..., \sigma_1))
\mathrm{d}\underline{\check{n}}'\\
&=&
\mathbb E( F(\underline{\check{n}}_\tau,\underline{\check{x}}) | \underline{\check{n}}_0 = \underline{\check{n}}, (\sigma_m, \sigma_{m-1}, ..., \sigma_0)).
\end{eqnarray*}
\end{proof}

\noindent
{\it Remarks.}  The duality statement above is a little more involved than that
for Markov processes with disjoint phase spaces 
(see e.g. Proposition 1.2  in \cite{JK14}).
Here, the phase spaces intersect in the set of particle configurations, and duality
is proved one sample path in particle movements at a time; that is why we call it
{\it pathwise duality}.
Note that it is necessary to use reversed 
paths in the dual process to guarantee 
that Lemma \ref{lemmaduality} can be applied from one step to the next. 
Also note that the expression "pathwise duality" has been used in
in Chapter 4 of \cite{JK14} in a different context:  
the "strong pathwise duality" and the "conditional pathwise duality"
as defined in \cite{JK14} imply the usual duality while our definition
is a weakening of that. Finally, we mention that we will have
to further weaken our concept of duality for the proof of the
case of systems with positive density of particles, 
see Section \ref{sec:dense}.

Duality has been used to prove LTE for a number of situations;
see \cite{CGGR13} and the references therein.
 Many of the ideas above have their origins in \cite{KMP82}, 
though modifications are needed as energy is not carried by particles in the KMP model.
A similar set of ideas was also used in \cite{RY07}, which considers
the same model as ours in one space dimension and with only one particle.

%%%%%%%%%%%%%%%%%%%%%%%%%%%%%%%%%%%%%%%
\subsection{Consequences of duality}

Let $\underline{\check{x}}_* \in \check{\mathcal X}$ and $\underline{\check{n}}_* \in \check{\mathcal Y}$ be fixed.  Integrating over all sample paths $\bm \sigma$ of
$\bar X_t$ on $[0,t]$, Lemma \ref{lemmaduality2} together with the
fact that $\mathbb P[\bm \sigma] = \mathbb P[\bm \sigma^{-1}]$ imply that
\begin{equation} \label{finite-t}
\int \mathbb E( F(\underline{\check{n}}_*, \underline{\check{x}}_t) | \underline{\check{x}}_0=\underline{\check{x}}_*, \bm \sigma) \mathbb \ \mathbb P(\mathrm{d}\bm \sigma)
= \int \mathbb E( F(\underline{\check{n}}_t, \underline{\check{x}}_*) | \underline{\check{n}}_0=\underline{\check{n}}_*, \bm \sigma^{-1}) \ \mathbb P(\mathrm{d}\bm \sigma^{-1})\ .
\end{equation}
Letting $t \to \infty$, the left and right sides of the equation above tend to the two
sides of the formula below:

\begin{lemma}
\label{lemma:dualitycons}
For any fixed %$\underline{\check{x}} \in \check{\mathcal X}$ and 
$\underline{\check{n}}_* \in \check{\mathcal Y}$, we have
\begin{equation}
\label{eq:lemmadc}
 \int  F(\check{\underline{n}}_*, \check{\underline{x}}) \mathrm{d}\mu^{(L)}(\check{\underline{x}}, \bar X)= 
\int \prod_{v \in \mathcal B_L} \left[ T\left( \frac{v}{L} \right) \right]^{\hat n_v} \mathrm{d}\rho_{\underline{\check n}_*},
\end{equation}
where $\rho_{\underline{\check n}_*}$ is the asymptotic distribution of ${\bm Y_t}$ 
as $t\rightarrow \infty$ averaged over all $\bar Y$ with uniform distribution, 
assuming that ${\bm Y_0} = (\underline{\check{n}}_*, \bar Y)$. 
\end{lemma}

That the left side of (\ref{finite-t}) converges to the left side of (\ref{eq:lemmadc})
as $t \to \infty$ follows from the fact that the distribution of $\bm X_t$ converges 
to $\mu^{(L)}$ (Proposition 2).
The limit on the right side clearly exists, since all  packets
are eventually deposited in $\mathcal B_L$, resulting in the simplified form of
$F(\check{\underline{n}}, \check{\underline{x}}_*)$ in the integrand. 
Notice also that the right side
does not depend on $\underline{\check{x}}_*$, consistent with the fact
that the convergence to $\mu^{(L)}$ on the left is independent of initial condition. 

We now identify the relevant choices of $\check{\underline{n}}_*$:
In the proof of LTE for site energies, for example, we 
fix $x \in \mathcal D$, $S \subset \mathbb Z^d$ 
and nonnegative integers $(n^*_v)_{v \in S}$.
If $\check{\underline{n}}_*$ is chosen so that 
\begin{equation}
\label{eq:n*choice}
n_{\langle xL \rangle + v} = \left\{ \begin{array}{rl}
 n^*_v &\mbox{ if $v \in S$} \\
  0 &\mbox{ if $v \not \in S$}
       \end{array} \right. \quad 
\hat n_v = 0 \ \ \forall v, \quad \text{ and  }\quad \tilde n_j = 0  \ \ \forall j\ ,
\end{equation} 
then the left side of (\ref{eq:lemmadc}) is equal to 
\begin{equation}
 \label{eq:moments}
\int \prod_{v \in S} \frac{ \xi_{\langle xL \rangle +v}^{n^*_v} }{n^*_v !} \mathrm{d}\mu^{(L)}\ ,
\end{equation}
a constant times the $(n^*_v)$-moments of the distribution $\mu^{(L)}_{x,S}$ 
defined in Sect. 2.2.

Thus the key to understanding $\mu^{(L)}$ is $\rho_{\underline{\check n}_*}$.
To get a handle on this distribution, we find that instead of working with $\bm Y_t$,
which describes the evolution of the density (or distribution) of packets, it is
productive to switch to an equivalent model that focuses directly on the movements
of individual packets. Moreover, since only asymptotic distributions matter,
we may work with a discrete time model, as long as the order of the steps are
preserved. 

\bigskip \noindent
{\bf Discrete-time version of $\bm Y_t$ focusing on movements of packets}

\medskip

Consider a Markov chain ${\bm Z_k}=
{\bm Z_k}^{(L)}, \ k=0,1,2,\dots$, the
variables of which are
$$
\bm Z_k = ({\bm Z}_{1,k}, \dots, {\bm Z}_{N,k}; Y_{1,k}, \dots, Y_{M,k})
\in  (\mathcal D_L \uplus \mathcal B_L \uplus \{ 1,2,...,M\})^N \times 
\mathcal D_L^M   \ ,
$$
with some fixed positive integer $N$ (to be specified later)
and the notation $\uplus$ for disjoint union.
The first $N$ coordinates of $\bm Z_k$ describe the positions of the $N$ (named) packets 
in the system, the position of packet $i$ at step $k$ being ${\bm Z}_{i,k}$,
and the final $M$ coordinates give the positions of the $M$ particles (abusing notation
slightly by using $Y$ in both the continuous and discrete time models).
The meaning of ${\bm Z}_{i,k} \in \mathcal D_L \uplus \mathcal B_L$
is obvious, and ${\bm Z}_{i,k} =j$ means packet $i$ is carried by particle $j$
 at time $k$. The transition probabilities of $\bm Z_k$ are as follows: 
Given $\bm Z_{k}$, we choose with equal probability one of the $M$ particles,
say particle $j$, and choose with equal probability one of particle $j$'s neighboring sites,
say $w$. If $w \in \mathcal D_L$, then we set $Y_{j,k+1}=w$, and mix the 
packets carried by particle $j$ with those at site $w$ by pooling them together 
and designating a random fraction of them to be carried by particle $j$ and 
the rest to be left at site $w$.
If $w \in \mathcal B_L$, then all the packets carried by particle $j$ are dropped off at site
$w$, and $Y_{j,k+1}=Y_{j,k}$. Since all the packets are eventually
dropped off in $\mathcal B_L$, $\bm Z_{i, \infty}:= \lim_{k \to \infty} \bm Z_{i,k}$ exists 
for $1 \leq i \leq N$ almost surely. 

Let $x \in \mathcal D, S \subset \mathbb Z^d$, and $L \gg 1$ be fixed. 
Associated with each 
$$\check{\underline{n}}_* = ((n^*_v)_{v \in S}, (\hat n^*_v)_{v \in \mathcal B_L},
\tilde n^*_1, \dots, \tilde n^*_M)$$ is the Markov chain $\bm Z_k$
whose initial condition $\bm Z_0=(\check{\underline{n}}_0, \bar Y_0)$ is given
by the following:
$\check{\underline{n}}_0$ is
prescribed by $\check{\underline{n}}_*$, i.e. at time $0$, there are $n^*_v$ 
packets at site $\langle xL \rangle + v$ and $\hat n^*_v$ packets at site 
$v \in \mathcal B_L$,
$\tilde n^*_j$ packets are carried by particle $j$, and $\bar Y_0$ is uniformly 
distributed among all particle configurations in $\mathcal D_L$. 
We claim -- and leave it to the reader to check -- that $\bm Y_t$ 
and $\bm Z_k$ differ only in the identity of individual
packets and time changes that preserve the order of the moves, so they have the
same asymptotic distribution, i.e. 
\begin{eqnarray}
\label{eq:dualityend}
 \int \prod_{v \in \mathcal B_L} \left[ T\left( \frac{v}{L} \right) \right]^{\hat n_v} d\rho_{\underline{\check n}_*}
=\mathbb E\left( \prod_{i=1}^N T \left( \frac{\bm Z_{i, \infty}}{L}\right) \right)
\end{eqnarray}
where $\mathbb E$  is with respect to the evolution of the
process $\bm Z_k$ and
$$N= \sum_{v \in S}n^*_v + \sum_{v \in \mathcal B_L} \hat n^*_v + 
\sum_{j=1}^M \tilde n^*_j\ .$$ 
For future reference, 
\begin{equation} \label{exp-exp}
\mathbb E(\cdot) = \int \mathbb E_{(\check{\underline{n}}_*,\bar Y)}(\cdot) \ U(\mathrm{d}\bar Y)
\end{equation}
where $U$ is the uniform distribution 
over all particle configurations $\bar Y$.

We are now ready to prove Theorem 1.

\begin{proof}[Proof of Theorem 1] Let $\check{\underline{n}}_*$ be such that
\begin{equation*}
%\label{eq:n*choice}
n_v = \left\{ \begin{array}{rl}
 1 &\mbox{ if $v = \langle xL \rangle$} \\
  0 &\mbox{ otherwise }
       \end{array} \right. \quad \text{ and  }\quad
\hat n_v = 0 \mbox{ for all } v, \quad \tilde n_j = 0  \mbox{ for all } j\ .
\end{equation*} 
Then the left side of (\ref{eq:lemmadc}) is equal to $\mathbb E^{(L)}(\xi_{\langle xL \rangle})$, and the right side is given by 
$\int T\left( \frac{v}{L} \right) \mathrm{d}\rho_{\underline{\check n}_*}$, where
$\rho_{\underline{\check n}_*}$ is the asymptotic distribution of 
$\bm Z_{1,\infty}$ in the Markov chain above with $N=1$ and $\bm Z_{1,0} = \langle xL \rangle$. From the transition probabilities of $\bm Z_k$, it is clear that if we (i) disregard 
waiting times, i.e. times at which $\bm Z_{1,k}$ does not change, and (ii) view
the location of the packet when it is carried by particle $j$ as $Y_{j,k}$, then
the trajectories of $\bm Z_{1,k}$ are those of a SSRW on $\mathcal D_L$. 
By Proposition \ref{lem12}, as $L \to \infty$ the distribution of
$\bm Z^{(L)}_{1,\infty}$ rescaled back to $\partial \mathcal D$ is the hitting probability
of Brownian motion starting from $x \in \mathcal D$.
Hence $\int T\left( \frac{v}{L} \right) \mathrm{d}\rho_{\underline{\check n}_*} \to u(x)$ 
where $u$ is the solution of
Laplace's Equation with boundary condition $T$. 
%We have proved 
%$\mathbb E^{(L)}(\xi_{\langle xL \rangle}) \to u(x)$ as $L \to \infty$. 
\end{proof}

Next we observe that Theorem 2 is reduced to the following proposition,
the proof of which will occupy the next two sections.

\begin{proposition}
 \label{lemma:main}
Let $d \ge 1$, $\mathcal D \subset \mathbb R^d$, $T$ on $\partial \mathcal D,\ 
x \in \mathcal D$ and $S \subset \mathbb Z^d$ be prescribed. 
We fix also $\check{\underline{n}}_* =((n^*_v)_{ v \in S}, 0, 0, \dots, 0)$, and let 
$\bm Z^{(L)}_k$ be the Markov chain associated with 
$\check{\underline{n}}_*$.
Then letting $N= \sum_{v\in S} n^*_v$, we have
%
%
%and $(n^*_v)_{v \in S}$, if $\check{\underline{n}}_*$ is such that $n_{\langle  xL \rangle + v}=n^*_v$ for $v \in S$ and all other coordinates are $0$, then 
$$ \lim_{L \rightarrow \infty}  \mathbb E\left( \prod_{i=1}^N T \left( \frac{\bm Z^{(L)}_{i, \infty}}{L}\right) \right) = [u(x)]^N\ .$$
\end{proposition}

\begin{proof}[Proof of Theorem 2 assuming Proposition \ref{lemma:main}] 
We first prove the result assuming tightness of the sequence 
$\left( \mu^{(L)}_{x,S} \right)_{L=1,2,...}$. Let 
$\mu_\infty$ be a weak limit point. Then putting together
 (\ref{eq:lemmadc}), (\ref{eq:dualityend}) and 
Proposition \ref{lemma:main}, we see that the moments of
$\mu_\infty$ are those of a product of exponential distributions with parameter
$\beta(x)=u(x)^{-1}$. Hence  $\mu_\infty$ is such a product; see Sect. 3.2. 
Since this is true for all limit points of $\mu^{(L)}_{x,S}$, we conclude 
that $\mu^{(L)}_{x,S}$ converges weakly to the measure claimed.

It suffices to prove tightness one coordinate at a time, so we may
assume $S=\{v\}$. Then the same reasoning (with $n^*_v=2$ in Proposition \ref{lemma:main}) implies
that $\sup_L \mathbb E^{(L)}(\xi_{\langle  xL \rangle + v}^2) \le C$ for some 
$C< \infty$.
Chebyshev's inequality then gives
$\mathbb P^{(L)}(\xi_{\langle  xL \rangle + v} >n) < \frac{C}{n^2}$ for all $L$, proving
tightness.
%
%
% of the sequence 
%$\left( \mu^{(L)}_{x,S} \right)_{L=1,2,...}$. Then via (\ref{eq:lemmadc}), (\ref{eq:dualityend}) and 
%Proposition \ref{lemma:main}, we see that the moments of
%$\mu_\infty$ are those of a product of exponential distributions with parameter
%$\beta(x)=u(x)^{-1}$. Hence  $\mu_\infty$ is such a product; see Sect. 3.2. 
%Since this is true for all limit points of $\mu^{(L)}_{x,S}$, it remains to show tightness.
%Note that the convergence of second moments (in fact their boundedness)
%combined with Chebyshev's inequality gives tightness of the coordinates, whence 
%tightness of the vector follows.
%\marginpar{I modified the proof of tightness}
%We conclude that $\mu^{(L)}_{x,S}$ converges weakly to the measure claimed.
\end{proof}

We close this section with the following lemma.

\begin{lemma}
\label{lemma:subsystem} Let ${\bm Z_k}$ be a system with $N$
packets, ${\bm Z^+_k}$ be another system with $N+1$ packets, and suppose
both have the same number of particles. Assume further that 
$\bm Z^+_{i,0} = \bm Z_{i,0}$ for $i=1,\dots, N$, 
and $Y^+_{j,0}=Y_{j,0}$ for all $j$.
Then with $\bm Z_k$ coupled to the corresponding coordinates in $\bm Z^+_k$
 in the natural way, 
we have $\bm Z^+_{i,k}=\bm Z_{i,k}$ for all $i=1, \dots, N$ and  $k \ge 1$.
\end{lemma}

\begin{proof} Without loss of generality, suppose that at step $k$, particle $1$ jumps from
site $z$ to site $w$, and that the union of the packets carried by this particle
or at site $w$ prior to the mixing are labelled $\{1, \dots, n\}$. Then
the probability that after the mixing, the set of packets carried by particle $1$ is
exactly $\{j_1, ... , j_l\} \subset \{1,2,...,n\}$ is given by
$$ p(j_1, ..., j_l) = \frac{1}{n+1} \frac{1}{{n \choose l}} = \frac{l!(n-l)!}{(n+1)!}$$
If packet $N+1$ is not at site $w$, then clearly the situation is not disturbed.
If it is there, we compute
$$ p(j_1, ..., j_l) + p(j_1, ..., j_l,N+1) = \frac{1}{n+2} \frac{1}{{{n+1} \choose l}} +\frac{1}{n+2} \frac{1}{{{n+1} \choose {l+1}}} = \frac{l!(n-l)!}{(n+1)!}\ .$$
Hence the dynamics of the first $N$ packets are unaffected. They are also 
clearly unaffected if particle $1$ drops off its packets at the bath. 
\end{proof}

\medskip
\noindent {\it Remark.}
An implication of Lemma \ref{lemma:subsystem} is that 
when $N>1$, the motion of each individual packet, when seen in the light of (i) and (ii) 
in the proof of Theorem 1, is a SSRW. Thus  
Proposition \ref{lemma:main} is proved if these SSRW are {\it independent},
or close enough to being independent. This is what we will show.

%%%%%%%%%%%%%%%%%%%%%%%%%%%%%%%%%%%%%%%%%
%%%%%%%%%%%%%%%%%%%%%%%%%%%%%%%%%%%%%%%
\medskip
\section{LTE for site energies: $d\ge 2$}
\label{sec:d=2}

In Sect. 5.1, we introduce, mostly for convenience, a small modification 
of the process $\bm Z_n$. This modified process is used a great deal in 
the pages to follow.
Sect. 5.2 contains the proof of LTE for site energies (Theorem 2) for $d=2$.
Due to the transience of SSRW, proofs for $d>2$ are simpler and
 are given in Sect. 5.3, along with the proof of Theorem 4.

\subsection{A slightly modified process} 

We have seen in the proof of Theorem 1 that with a suitable modification of
$\bm Z_n$, the movement of the packet becomes a SSRW. We now carry out
the same type of modification systematically under more general conditions:
%
%it is profitable to focus on 
%the movements of {\it packets} (rather than the movements of particles). We now 
%propose a systematic modification of the process $\bm Z_n$ to put it in
%a form more convenient for our purposes:

\medskip
The phase space of $\hat{\bm Z}_n$ is $(\mathcal D_L \cup \mathcal B_L)^N$,
and its dynamics are derived from those for $\bm Z_n$ in the following way:
First, let $\bm Z'_{i,n}= \pi(\bm Z_{i,n})$
where $\pi(\bm Z_{i,n})=\bm Z_{i,n}$  if $\bm Z_{i,n} \in \mathcal D_L 
\cup \mathcal B_L$, and $\pi(\bm Z_{i,n})=Y_{\ell,n}$ if $\bm Z_{i,n}=\ell$. 
We then let $t_0=0$, and for $j=1,2,\dots$, define
$$
t_j = \min_{n>t_{j-1}} \{ \exists i \in \{1,2, ..., N \}: \bm Z'_{i,n} \neq \bm Z'_{i,n-1}\}\ .
$$
Finally, set $\hat{ \bm Z}_{i,n}= \bm Z'_{i,t_n}$. That is to say, first we confuse 
being at a site and being carried by a particle at that site, and then we collapse
the times when there is no action according to this way of bookkeeping.

\bigskip \noindent
{\it Remark.} We recognize that $\hat{\bm Z}_n$ is not Markovian (and is not
especially nice as a stochastic process). However, 
the {\it order} in which the $N$ 
packets move about on $\mathcal D_L \cup \mathcal B_L$ is preserved as we
go from $\bm Z'_n$ to $\hat{\bm Z}_n$, even as time has been reparametrized. 
As a consequence, $\hat{\bm Z}_n$ has the following important properties:
\begin{itemize}
\item[1.] For each $L$, the joint asymptotic distribution of 
$(\hat{\bm Z}^{(L)}_{1, \infty}, \dots, \hat{\bm Z}^{(L)}_{N, \infty})$
is identical to that of $(\bm Z^{(L)}_{1, \infty}, \dots, \bm Z^{(L)}_{N, \infty})$.
\item[2.] Each packet individually performs a SSRW on 
$\mathcal D_L \cup \mathcal B_L$
modulo waiting times (during which it stands still).
\item[3.] The addition of new packets in the sense of Lemma 9 does not affect
the order of movements of packets already under consideration.
\end{itemize}

\smallskip
When two packets are at the same site, their next moves are not independent.
We prove a uniform bound on how long they are likely to stick together:

\begin{lemma}
\label{lemma:sticking} Assume $N=2$, and  
$\hat{\bm Z}_{1,k_0} = \hat{\bm Z}_{2, k_0} \not \in \mathcal B_L$. 
Let $\kappa$ be the smallest positive
integer such that $\hat{\bm Z}_{1,k_0+\kappa} \ne \hat{\bm Z}_{2, k_0+\kappa}$.
Then
$$
\mathbb P(\kappa > k) \le  (2/3)^{\frac{k-1}{2}} \qquad k=1,2, \dots.
$$
\end{lemma}

\begin{proof}  The only way to find out what happens in the $\hat{\bm Z}$-process
is to go back to the corresponding step in the $\bm Z$-process. Below we 
enumerate all possible scenarios for $\bm Z_{i,t_{k_0}}, i=1,2$, that
correspond to $\hat{\bm Z}_{1,k_0} = \hat{\bm Z}_{2, k_0}$, and consider 
for each scenario the probability of the two packets staying together in the
next one or two steps:

\smallskip \noindent
{\it Scenario 1.} At time $k_0$, exactly one of the packets is carried by a particle, 
or the two packets
are carried by different particles. In both cases, $\kappa=1$, i.e., they will separate
in the next step. 

\smallskip \noindent
{\it Scenario 2.} Both packets are carried by the same particle. Then 
$\mathbb P(\kappa=1) = 0$ but $\mathbb P(\kappa=2) \ge \frac13$. 
Reason: This particle jumps, carrying both packets to the next site, where
with probability $\frac13$ it drops one packet and carries the other,
a scenario that is guaranteed to lead to $\kappa=2$.

\smallskip \noindent
{\it Scenario 3.} Neither packet is carried by a particle. Before the next move can
occur, a particle has to enter the site,
and with probability $(\frac13,\frac13,\frac13)$, picks up (i) neither, (ii) one, 
or (iii) both of the packets. If (i) occurs, Scenario 3 is repeated. 
(ii) and (iii) are followed by Scenarios 1 and 2 
respectively.  
\end{proof}

{\it Notation:} In this paper, we denote every universal constant by $C$, so that 
each occurrence of $C$ may stand for
a different number, even in the same line.

%%%%%%%%%%%%%%%%%%%%%%%%%%%%%%%%%%

\subsection{Proof of Theorem 2: $d=2$}

We now focus on the planar case. Let $\bm Z_n=\bm Z_n^{(L)}$ be 
as in Proposition 4,  and
$\hat{\bm Z}_n$ the modification of $\bm Z_n$ as defined in Sect. 5.1.
For $\beta, \delta \in (0,1)$ and $i,j \in \{1,\dots, N\}$, we define
\begin{eqnarray*}
\tau_{i,j} & = &
\min \{ n: \dist (\hat{\bm Z}_{i,n}, \hat{\bm Z}_{j,n}) > 
L^{\beta}\} \\
\mbox{ and } \qquad  \mathcal T_i & = &
\min \{ n: \dist (\hat{\bm Z}_{i,n}, \langle xL \rangle  )>
L^{\beta + \delta}\}. 
\end{eqnarray*}
In general, the definitions of $\tau_{i,j}$ and $\mathcal T_i$ depend on packets
other than $i$ and $j$ (due to the way we collapse time when going from $\bm Z_n$
to $\hat{\bm Z}_n$), so let us first assume these are the only two packets present.

\begin{lemma}
\label{lemma3} Consider $\mathcal D \subset \mathbb R^2$, and assume $N=2$. 
Then for every $\beta \in (0,1)$ and $\delta \in (\frac{1-\beta}{3},1-\beta)$ there is a  constant $C= C(\beta, \delta)$ such that
$$ \mathbb P(\mathcal T_1 <\tau_{1,2}) <\frac{C}{L^{100}} \qquad
\mbox{ for all sufficiently large } L\ .
$$
% $$P \left( \| \hat{\bm Z}_{i, 0} - \hat{\bm Z}_{i,\tau_{i,j}} \| > L^{\beta + \delta} \right) < \frac{C}{L^{100}}.$$
\end{lemma}

\begin{proof} Our plan is to write
\begin{eqnarray*} \label{T1}
\mathbb P \left( \mathcal T_1 <\tau_{1,2} \right)
&<&  \mathbb P \left( \mathcal T_1 <\tau_{1,2},  \tau_{1,2} <  CL^{1 + \beta - \delta} \right)
 + \mathbb P \left(\tau_{1,2} >  CL^{1 + \beta - \delta} \right)\\
& < & \mathbb P \left( \mathcal T_1<  CL^{1 + \beta - \delta} \right) + 
\mathbb P \left(\tau_{1,2} >  CL^{1 + \beta - \delta} \right)\ ,
\end{eqnarray*}
and prove that each of the two terms above is $<\frac{C}{L^{100}}$ for large enough $L$.

Consider first the second term. Decomposing the steps before $\tau_{1,2}$ 
according to whether $\hat{\bm Z}_{1,n} = \hat{\bm Z}_{2,n}$, we claim that 
\begin{equation}
\label{eq:waspinequality}
 \mathbb P(|\{ n:  \hat{\bm Z}_{1,n} \neq \hat{\bm Z}_{2,n}, n \leq \tau_{1,2} \}| > L^{1+\beta - \delta}) < \frac{C}{L^{100}},
\end{equation}
where $|.|$ denotes the cardinality of a set. To see that,
we let $a_0 = 0$, and define
$$ a_n=
\min \{k > a_{n-1} \,|\, \mathbf{\hat{Z}}_{1,k} \neq
 \mathbf{\hat{Z}}_{2,k}\} \qquad \mbox{ and } \qquad 
U_{n} = \mathbf{\hat{Z}}_{1,a_{n}} -
\mathbf{\hat{Z}}_{2,a_{n}}\ .
$$
Then the process $U_{n}$ is a planar SSRW as long as it is away from the
origin. Whenever the SSRW would reach the origin, $U_n$ performs
two steps of the SSRW thus avoids the origin (more precisely,
for any $x \in \mathbb Z^2$ and any $e_i,e_j \in 
\{ (0,1),(0,-1),(1,0),(-1,0)\}$,
$\mathbb P(U_{n+1} = x+ e_i | U_n = x) = 1/4$ if $x + e_i \neq (0,0)$
and
$\mathbb P(U_{n+1} = e_j | U_n = e_i) = 1/16$). Observe that by construction
$$
  | \{n \,|\, \mathbf{\hat{Z}}_{1,n} \neq \mathbf{\hat{Z}}_{2,n}; n
  \leq \tau_{1,2}\} | = \min_{n} \{ \|U_{n}\| > L^{\beta}\}\ .
$$
It is easy to show by the invariance principle that there is a $p>0$ such that
$$
\mathbb  P( \min_{n} \{ \|U_{n} \| \geq L^{\beta}\} < L^{2\beta}) \geq p
$$
holds for any starting position $U_0$ with $\| U_0\| < L^{\beta}$. In fact, 
the left hand side converges to $e^{-1/2}$ as $L \rightarrow \infty$, but we will not need this.
%Since $ |\hat{\bm Z}_{1,n} - \hat{\bm Z}_{2,n}|$ is a one dimensional random walk whenever it is positive, it has a positive 
%probability (say $p$) of reaching $L^{\beta}$ in $L^{2 \beta}$ time. 
By induction, $\| U_n\|$ reaches $L^{\beta}$ in time $kL^{2 \beta}$ with 
probability at least $1-(1-p)^k$. Since $1-\beta-\delta >0$, the choice $k = L^{1-\beta - \delta}$ gives (\ref{eq:waspinequality}).

Now it is well known that for two dimensional random walks, the number of returns to the origin up to time $L^{1+\beta-\delta}$ is $O(\log L)$ as $L \to \infty$. 
Furthermore, formula (3.11) in \cite{ET60} implies that 
the probability of the number of returns to the origin being bigger than 
$\lceil \frac{100}{\pi} \rceil \log^{2} L$ is bounded by $C/L^{100}$. Letting
$\varsigma_k$ be the duration that the two packets stick together
at their $k$th meeting, we note that these random variables are independent, 
and each is stochastically bounded by the geometric distribution in Lemma \ref{lemma:sticking}. Thus we have
\begin{eqnarray*}
& & \mathbb P \left(|\{n \leq \tau_{1,2} : \mathbf{\hat{Z}}_{1,n} =\mathbf{\hat{Z}}_{2,n}\} | >  CL^{1 + \beta - \delta} \right)\\
 &< & \mathbb P \left(\# \mbox{ meetings before } \tau_{1,2}  > \lceil \frac{100}{\pi} \rceil  \log^{2} L \right)
\ + \ 
 \mathbb P\left(\sum_{m=1}^{\lceil \frac{100}{\pi} \rceil \log^{2} L} \varsigma_m > C L^{1+ \beta - \delta}\right)\\
 & < & \frac{C}{L^{100}} +  \frac{C}{L^{100}}\ .
\end{eqnarray*}
This completes the proof of $\mathbb P \left(\tau_{1,2} >  CL^{1 + \beta - \delta} \right)
< \frac{C}{L^{100}}$.

To finish, we let $W_n$ denote the SSRW corresponding to $\hat{\bm Z}_{1,k}$ with
waiting times collapsed, and note that
\begin{eqnarray*}
\mathbb P \left( \mathcal T_1<  CL^{1 + \beta - \delta} \right) & < & 
\mathbb P \left(\sup_{n<CL^{1 + \beta - \delta}} \|\sum_{k=1}^n W_k \| > L^{\beta + \delta} \right)\\
& < & \sum_{n=1}^{CL^{1 + \beta - \delta}} \mathbb P \left(\| \sum_{k=1}^n W_k \| > L^{\beta + \delta} \right)\ .
\end{eqnarray*}
Since our assumption $(1-\beta)/3 < \delta$ implies $ \beta + \delta > 
\frac12 (1 + \beta - \delta)$, it follows by moderate deviation theorems such as
Lemma \ref{lemma2} (by projecting onto one coordinate, for example) that
each of the probabilities in the last line is bounded above by a term of the form
$< c_1 e^{-c_2 L^{c_3}}$.
Thus the sum is $<\frac{C}{L^{100}}$ for large $L$. 

The proof of the lemma is now complete.
\end{proof}

\smallskip
We now return to the case of $N$ packets for arbitrary $N$, and let 
$\mathcal T_i $ and $\tau_{i,j}$ be as defined at the beginning of
this subsection.

\begin{lemma} Let  $\mathcal D \subset \mathbb R^2$. Given any $\varepsilon>0$, let
$\beta = 1-2\varepsilon$ and $\delta=\varepsilon$. Then for any 
 $i \in \{1, \dots, N\}$, the following holds for all sufficiently large $L$: 
\begin{equation} \label{Ti}
 \mathbb P (\hat{\bm Z}_{i,n} = \hat{\bm Z}_{j,n} \in \mathcal{D}_{L} \mbox{ for some } n> 
\mathcal T_i \mbox{ and some } j \ne i) < 4(N-1)\varepsilon\ .
\end{equation}
\end{lemma}

\begin{proof} First we claim that for each
$i,j \in \{1, \dots, N\}$ with $i \ne j$, we have
\begin{equation} \label{tauij}
\mathbb P(\mathcal T_i <\tau_{i,j}) <\frac{C}{L^{100}}\ ,
\end{equation}
i.e., Lemma 11 in fact holds for any pair $i,j$
in a process with $N$ packets. To go from $N=2$ to general $N$,
observe that while the definitions of $\mathcal T_i$ and $\tau_{i,j}$
depend on the packets present (more time steps are collapsed when there
are fewer packets), (\ref{tauij}) concerns only the {\it relation} between 
$\tau_{i,j}$ and $\mathcal T_i$, not the actual values of these random variables,
and relations of this type are not affected by the presence of other packets
as noted in the Remark in Sect. 5.1. 

Now for fixed $j \ne i$, Lemma \ref{lemma:mesoscop} applied to 
$W_n=\hat{\bm Z}_{i,n} - \hat{\bm Z}_{j,n}$ 
tells us that with probability $>1-3\varepsilon$, packets $i$ and $j$ do not meet 
after $\tau_{i,j}$ for $L$ large. (The result in Lemma 4 is not affected by waiting times.)
Thus picking $L$ large enough so the right side of (\ref{tauij}) is $<\varepsilon$, we have
\begin{eqnarray*}
& & \mathbb P (\hat{\bm Z}_{i,n} = \hat{\bm Z}_{j,n} \in \mathcal{D}_{L} \mbox{ for some } n> 
\mathcal T_i )\\
& < & \mathbb P(\mathcal T_i  < \tau_{i,j}) + 
 \mathbb P (\hat{\bm Z}_{i,n} = \hat{\bm Z}_{j,n} \mbox{ for some } n> 
\mathcal T_i | \mathcal T_i  > \tau_{i,j})\\
& < & \varepsilon + 3 \varepsilon = 4 \varepsilon\ .
\end{eqnarray*}
Summing over all $j \ne i$ gives (\ref{Ti}).
\end{proof}

For $d=2$, we  now prove Proposition 4, from which Theorem 2 follows as explained in
Sect. 4.3.

\begin{proof}[Proof of Proposition 4 ($d=2$):] We will prove 
\begin{equation}\label{indep}
\left|\mathbb E \left( \prod_{j=1}^N T \left( \frac{\hat{\bm Z}^{(L)}_{j, \infty}}{L}\right) \right) -  u(x)^N \right| \to 0 \qquad \mbox{ as } L \to \infty
\end{equation}
inducting on $N$. The case of $N=1$ is Theorem 1. We now assume (\ref{indep})
has been proved for a process with $N-1$ packets, and note that when embedded
in a process with $N$ packets, the same asymptotic distribution holds for packets
$1,2,\dots, N-1$. 

Let an arbitrarily small $\varepsilon>0$ be fixed. We consider $L$ large enough
for (\ref{Ti}) to hold, and for such $L$, we define the stopping time 
$\mathcal S$ to be the smallest 
$n > \mathcal T_N$ such that $\hat{\bm Z}_{N,n}=\hat{\bm Z}_{j,n} \in \mathcal{D}_{L} $ for
some $j \ne N$ (setting $\mathcal S=\infty$ otherwise),
and define another process $\hat{\bm Z}^*_{N,n}$ with the 
property that $\hat{\bm Z}^*_{N,n}=\hat{\bm Z}_{N,n}$
for $n \le \mathcal S$, and  it is a SSRW independent
of the movements of the other $N-1$ packets after time $\mathcal S$. 
This ensures (i) $\hat{\bm Z}^*_{N,\infty}$
is independent of the joint distribution of $\hat{\bm Z}_{j,\infty}, j=1,\dots, N-1$,
and (ii) $\mathbb P(\hat{\bm Z}^*_{N,\infty} \ne \hat{\bm Z}_{N,\infty}) < 
4(N-1) \varepsilon$ for all large enough $L$.

Simplifying notation by dropping the superscript$\ ^{(L)}$ and abbreviating
$T \left( \frac{\hat{\bm Z}_{j, \infty}}{L}\right)$ as $T(z_j)$ and
$T \left( \frac{\hat{\bm Z}^*_{N, \infty}}{L}\right)$ as $T(z^*_N)$, we have
\begin{eqnarray*}
 |\mathbb E (\Pi_{j=1}^N T (z_j)) - u(x)^N |
& \le & | \mathbb E ((\Pi_{j=1}^{N-1} T (z_j)) \cdot ( T( z_N) - T( z^*_N) )) |\\
& & \quad + \  |\mathbb E (\Pi_{j=1}^{N-1} T (z_j))  \cdot (\mathbb E(T(z_N^*)) - \mathbb E(T(z_N)) ) |\\
& &  \quad + \  |\mathbb E (\Pi_{j=1}^{N-1} T (z_j))  \cdot (\mathbb E(T(z_N)) - u(x)) |\\
& & \quad  + \ |\mathbb E (\Pi_{j=1}^{N-1} T (z_j))  \cdot u(x) - u(x)^N |\ .
\end{eqnarray*}
Of the 4 lines on the right side, the first and the second $= O(\varepsilon)$,
due to the fact that $T$ is bounded and 
$\mathbb P(T(z_N) \ne T(z^*_N)) = O(\varepsilon)$;
the third tends to $0$ as $L \to\infty$ by Theorem 1 and
  Proposition 3; and the fourth tends to $0$ as
$L \to \infty$ by our inductive hypothesis. This completes the inductive step and the proposition.
\end{proof}

%%%%%%%%%%%%%%%%%%%%%%%%%%%%%%%%%%%%%%
\subsection{Related proofs}

Details aside, the proof of the $d=2$ case of Theorem 2 can be summarized
 as follows: 
We fix two distinct length scales,
$L^{\beta + \delta}$ and $L^\beta$, and consider $\mathcal T_i$, the time it takes
packet $i$ to attain a net displacement of $L^{\beta + \delta}$, and $\tau_{i,j}$, the time
it takes packets $i$ and $j$ to separate by a distance of $L^\beta$. Notice that
$\hat{\bf Z}_{i,n} - \hat{\bf Z}_{j,n}$ is a SSRW except for the fact that the packets 
tend to stick together for a random time with finite expectation when they meet.
We then showed that

(a) with high probability, $\mathcal T_i > \tau_{i,j}$, due to the difference in 
length scale and 

\quad also to the fact that $O(\log L)$, the number of encounters between 
packets 

\quad $i$ and $j$ before $\tau_{i,j}$, is insignificant; 

(b) two packets that are $L^\beta$ apart  are not likely to meet.

\noindent It follows from (a) and (b) that after time $\mathcal T_i$,
the trajectories of packets $i$ and $j$ are effectively independent, and
the desired result follows from the harmonic measure characteristization
of hitting probability starting from $x \in \mathcal D$. 

\medskip
The proofs below follow the same argument, with some simplifications.

\begin{proof}[Proof of Proposition 4 ($d>2$)]  Transience of SSRW in $d>2$
simplifies the estimates. Specifically, let $W_n$ be a SSRW in $\mathbb Z^d$, $d>2$. Then the number of encounters in (a) can be estimated by the fact that 
given $\varepsilon>0$, there exists $K$ such that
$$\mathbb P(\# \mbox{ encounters } > K ) < \varepsilon \ ,
$$
and (b) follows from the fact that given $\varepsilon>0$, there exists
$C_1$ such that 
$$
\mathbb P(W_k =0 \mbox{ for some } k |W_0 > C_1 ) < \varepsilon\ .
$$
Indeed the two length scales in the proof in Sect. 5.2 can, if one so chooses, be 
replaced by two suitably related constants $C_1 < C_2$.
\end{proof}

\smallskip
\begin{proof}[Proof of Theorem 4 for $d \ge 2$] Theorem 4 differs from Theorem 2 
in that the initial locations of the packets may be 
$O(L^\vartheta)$ apart. As can be seen from the sketch of proof above,
the following two places in the argument may be affected: (i) With regard to
the number of encounters before $\tau_{i,j}$ in (a), the probability of meeting 
at least once cannot be increased if the packets are farther apart, and 
once they meet, the probability of meeting again is independent
of their initial separation. (ii) For each $L$, 
when rescaled back to $\mathcal D$ the packets do not start from
$x$ but from $x_j$ with $|x_j-x| = O(L^{\vartheta-1})$. The convergence
of hitting probabilities starting from these slightly perturbed initial conditions 
is covered by Proposition \ref{lem12}.
\end{proof}

%%%%%%%%%%%%%%%%%%%%%%%%%%%%%%%%%%%%%%%
%%%%%%%%%%%%%%%%%%%%%%%%%%%%%%%%%%%%%%%%%

\medskip
\section{LTE for site energies: $d=1$}
\label{sec:d1}

In dimension $1$, independent random walkers meet too often for the type
of argument in Section 5 to work. On the other hand, when two packets meet,
it suggests the possibility of exchangeability, and we will make use of this in
our proof. 
We may assume without loss of generality that $\mathcal D = [0,1]$,
so that $\mathcal B_L = \{-1,L+1\}$, and write $\mathcal L = -1, \mathcal R = L+1$.

\subsection{The case of $N=2$}

The argument in this subsection is borrowed from \cite{RY07}; we need only
the $N=2$ case, which is considerably simpler.
Starting from usual initial conditions, we let $A_{i} = \mathcal{L}$ or $\mathcal{R}$, 
and define
\begin{eqnarray*}
P_L(A_{1}, A_{2}) & = & \mathbb{P}( \mathbf{Z}^{(L)}_{1, \infty} = A_{1}, 
\mathbf{Z}^{(L)}_{2, \infty} = A_{2}) ,\\
 P_L(A_{1}, *) & = & \mathbb{P}( \mathbf{Z}^{(L)}_{1, \infty} = A_{1}), \qquad 
P(\cdot, \cdot) = \lim_{L \to \infty} P_L(\cdot, \cdot) , \mbox{ etc. } 
\end{eqnarray*}

\begin{lemma}
It is sufficient to show that
$$
P(\mathcal{R}, \mathcal{R}) = x^{2} \,.
$$
\end{lemma}
\begin{proof} We know from 1d SSRW  (or the gambler's ruin problem) that
\begin{equation}\label{N=1}
P(*, \mathcal{R}) = x
\qquad \mbox{ and } \qquad P( \mathcal{L},*)=1- x \,.
\end{equation}
If $P(\mathcal{R}, \mathcal{R}) = x^2$, then 
$P(\mathcal L, \mathcal{R}) = P(*, \mathcal{R}) - P(\mathcal{R}, \mathcal{R})
=x-x^2 = x(1-x)$, which is also equal to $P(\mathcal{R}, \mathcal L)$ by
symmetry. Thus as $L \to \infty$, $\bm Z_{i,\infty}$ and $\bm Z_{2,\infty}$ are
independent Bernoulli random variables with weights $(1-x, x)$, giving
$$
\lim_{L \rightarrow \infty} \mathbb{E}\left( T\left( \frac{\mathbf{{Z}}_{1, \infty}}{L} \right) T\left( \frac{
    \mathbf{{Z}}_{2, \infty}}{L} \right) \right)  =[u(x)]^2\,.
$$
\end{proof}

\begin{lemma}
\label{lem:lowerbd} \ 
$P(\mathcal{R}, \mathcal{R}) \geq x^{2 } $
\end{lemma}
\begin{proof}
Let $f( x, y) = xy$. For each $\hat{\bf Z}^{(L)}_n$-process as defined in Sect. 5.1, define
$$
  \Delta f(\mathbf{\hat{Z}}_{1, n}, \mathbf{\hat{Z}}_{2,n}) : = \mathbb{E}( f(
  \mathbf{\hat{Z}}_{1, n+1}, \mathbf{\hat{Z}}_{2, n+1}) \,|\, \mathbf{\hat{Z}}_{1, n},
  \mathbf{\hat{Z}}_{2, n}) - f( \mathbf{\hat{Z}}_{1,n}, \mathbf{\hat{Z}}_{2, n}) \,.
$$
Then $\Delta f = 0$ for all situations except when
$\mathbf{\hat{Z}}_{1, n} = \mathbf{\hat{Z}}_{2, n}$, in which case we have 
$$
  \mathbb{E}( \Delta f \,| \, \mbox{ two packets move together }) =
  \frac{1}{2}((i-1)^{2} + (i+1)^{2}) - i^{2}  = 1 \,,
$$ 
where $i$ is the location of the packets; otherwise $\Delta f  = 0$.

Therefore, for any finite $L$, $f(\mathbf{\hat{Z}}_{1, n}, \mathbf{\hat{Z}}_{2, n})$ is a bounded
submartingale. As $n \rightarrow \infty$, $f(\mathbf{\hat{Z}}_{1,
  n}, \mathbf{\hat{Z}}_{2, n})$ converges to $f(\mathbf{\hat{Z}}_{1, \infty},
\mathbf{\hat{Z}}_{2, \infty})$ almost surely. Furthermore, by the submartingale property, we have
$E( f(\mathbf{\hat{Z}}_{1, \infty},
\mathbf{\hat{Z}}_{2, \infty})) \geq 
E( f( \mathbf{\hat{Z}}_{1, 0},
\mathbf{\hat{Z}}_{2, 0}))$.
Thus
$$ P(\mathcal{R}, \mathcal{R})  =  \lim_{L \rightarrow \infty} L^{-2}\mathbb{E}( f(\mathbf{\hat{Z}}_{1, \infty},
\mathbf{\hat{Z}}_{2, \infty})) 
\geq \lim_{L \rightarrow \infty} L^{-2} f( \mathbf{\hat{Z}}_{1, 0},
\mathbf{\hat{Z}}_{2, 0}) = x^{2} \,.
$$
\end{proof}

\begin{lemma} \label{upbound} \ 
$P(\mathcal{R}, \mathcal{R}) \leq x^{2} $
\end{lemma}

The proof of this lemma is simpler if we work with a process that differs from
$\mathbf{\hat{Z}}_n$ by a half-step. More precisely, in $\bm Z_n$, a packet
first jumps before it picks up a random fraction of the packets at its destination.
We let $\bm Z^+_n$ be $\bm Z_n$ with the order of jumping and mixing reversed,
and apply the procedure at the beginning of Sect. 5.1 to $\bm Z^+_n$
to obtain a process we call $\mathbf{\tilde{Z}}_n$ (instead of $\mathbf{\hat{Z}}_n$). 
Clearly, as $L \to \infty$, $\mathbf{\hat{Z}}_n, \mathbf{\tilde{Z}}_n$ and $\bm Z_n$ all 
have the same asymptotic packet distributions. The notation
$P(\cdot, \cdot)$ below refers to $\mathbf{\tilde{Z}}_n$. 

\begin{proof}[Proof of Lemma 15]
We consider the function
$$
  f(x,y) = xy - c |x-y| 
$$
where $c$ is a constant to be determined. Let $\Delta f$ be as before. 
%
%As before, define
%$$
%  \Delta f(\mathbf{\tilde{Z}}_{1, n}, \mathbf{\tilde{Z}}_{2,n}) : = \mathbb{E}[ f(
%  \mathbf{\tilde{Z}}_{1, n+1}, \mathbf{\tilde{Z}}_{2, n+1}) \,|\, \mathbf{\tilde{Z}}_{1, n},
%  \mathbf{\tilde{Z}}_{2, n}] - f( \mathbf{\tilde{Z}}_{1,n}, \mathbf{\tilde{Z}}_{2, n}) \,.
%$$
%
Then it is easy to see that $\Delta f = 0$ when $\mathbf{\tilde{Z}}_{1, n } \neq 
\mathbf{\tilde{Z}}_{2, n}$. If $\mathbf{\tilde{Z}}_{1, n} = \mathbf{\tilde{Z}}_{2, n}$, 
then  $\Delta f =1$ if the two packets move together in the next step, and
$\Delta f = -c$  if the next step involves exactly one of the packets.
Now the two packets can move together only if they are both available
to be picked up at step $n+1$, and even in that case, the probability that
exactly one of them is picked up by the particle in question is $=\frac13$.
Thus we conclude that independently of what is going on in the $\bm Z^+$-process, 
\begin{equation} \label{prob}
\mathbb{P}(\mathbf{\tilde{Z}}_{1, n+1 } \neq \mathbf{\tilde{Z}}_{2, n+1} | 
\mathbf{\tilde{Z}}_{1, n } = \mathbf{\tilde{Z}}_{2, n}) \ge \frac13\ .
\end{equation}
Choosing $c> 3$ therefore will ensure that $\Delta f \le 0$ at each step.
Arguing as above, we then obtain $P(\mathcal{R}, \mathcal{R}) \leq x^{2}$.
\end{proof}

Notice that the left side of (\ref{prob}) can be zero
if we use $\mathbf{\hat Z}_n$ instead of $\mathbf{\tilde{Z}}_n$, for two packets
carried by the same particle will necessarily move together if the particle
first jumps before it mixes. This is our reason for using $\mathbf{\tilde{Z}}_n$.

%%%%%%%%%%%%%%%%%%%%%%%%%%%%%%%%%%%%%%
\subsection{Proofs of Theorems 2 and 4 in the case $d=1$}

Following \cite{KMP82}, we use the method of exchangeable random variables
to extend the results above to the case of $N$ packets for
arbitrary $N$. 

\medskip
Fix $x \in [0,1]$ and $S \subset \mathbb Z$. Let $v_1, v_2, \dots \in S$ (repeats
allowed). 
For each $N$ and $L$, we consider the process $\bm Z_n=\bm Z_n^{(L)}$ with 
$Z_{i,0}=\langle xL \rangle + v_i, i=1,2, \dots, N$, and 
for $A_j=\mathcal L$ or $\mathcal R$,
 we define
$$
p_{L,v_1, \dots, v_N}(A_1, \dots, A_N) = \mathbb P \left(\frac{\bm Z_{1,\infty}}{L} = A_1, 
\dots , \frac{{\bm Z}_{N,\infty}}{L} = A_N \right)\ ,
$$
where the probability is defined to be the average over all initial particle 
configurations $\bar Y$.
We begin with a lemma that sets the stage for exchangeability.

\begin{lemma}
\label{lemma:exch}
For every $N \in \mathbb Z^+$ and every permutation $\sigma$ 
of $\{1,\dots, N\}$, 
$$ 
\lim_{L \rightarrow \infty}\  \left(p_{L,v_1, \dots, v_N}(A_1, \dots, A_N) - p_{L,v_1, \dots, v_N}(A_{\sigma (1)}, \dots, A_{\sigma (N)}) \right) \ = \ 0.
$$
\end{lemma}

\begin{proof} Since binary transformations generate the symmetric group,
it suffices to consider $\sigma$ with the property that for some $i \ne j$,  
$\sigma (i)=j, \sigma (j)=i$, and $\sigma (\ell) = \ell$ for all $\ell \ne i,j$. 
For fixed $L$, we consider 
the process $\bm Z_n=\bm Z_n^{(L)}$, and begin with the following observation:

 Abbreviating $p_{L,v_1, \dots, v_N}$ as $p$, 
we claim that
$$
p(A_1, \dots, A_N | \bar Y, F) = 
p(A_{\sigma (1)}, \dots, A_{\sigma (N)} | \bar Y, F)
$$
if $F$ is one of the following two types of events:

\smallskip
(i) $F= \{\bm Z_{i,n}=\bm Z_{j,n}\}$ for some $n$;

(ii) for  two particles $k \ne k'$ and for some $n$, 
$$F=\{\bm Z_{i,n}=k, \bm Z_{j,n}=k',  Y_{k,n}=Y_{k', n} \mbox{ and } 
\bm Z_{\ell,n} \ne k, k' \ \forall \ \ell \ne i,j \}\ .
$$
\smallskip
\noindent To see that the asymptotic distributions are as claimed in case (i), 
we simply switch the roles of packets $i$ and $j$ from time 
$n$ on. In case (ii), packets $i$ and $j$ are carried by two different particles,
which are at the same site at time $n$. In this case, we switch not only the
roles of packets $i$ and $j$ but the sets of randomness for particles $k$ and $k'$
from this time on. The condition that these particles  do not carry
other packets at time $n$ ensures that 
the asymptotic distributions of other packets are not affected.
We will refer to an event corresponding to (i) or (ii) above as 
a ``viable switching" for packets $i$ and $j$.

\medskip
Let $\bm{\hat Z}_n$ be the process obtained from $\bm Z_n$. As explained
in Section 5, module waiting times, $\bm{\hat Z}_{i,n} - \bm{\hat Z}_{j,n}$ is a SSRW
when it is $\ne 0$. When $\bm{\hat Z}_{i,n} = \bm{\hat Z}_{j,n}$, the duration
they spend together following each encounter is controlled by Lemma 10. 
As these durations are independent for different encounters and 
are bounded by random variables with finite expectations, it follows from the
discussion in Sect. 3.1 that the number of times
packets $i$ and $j$ meet before reaching the bath
is $O(L)$; we in fact need only that 
\begin{equation}
 \label{eq:encounters}
\mathbb P ( \text{\# of encounters }> C) \rightarrow 1 \quad
\text{ as $L \rightarrow \infty$ for every $C$.}
\end{equation}

Let $\hat \tau_1< \hat \tau_2 < \dots$ be the times packets $i$ and $j$  meet
in $\mathcal D_L$ in the $\bm{\hat Z}_n$-process, i.e. $\bm{\hat Z}_{i, \hat \tau_q}=
\bm{\hat Z}_{j, \hat \tau_q} \in \mathcal D_L$, and 
let $\tau_q$ denote the corresponding times in the
$\bm Z_n$-process. We let $F_q$ denote
the event that a viable switching occurs at the $q$th meeting, and
will show that there exists $b=b(N)>0$ such that
\begin{equation} \label{switchprob}
\mathbb P(F_{q} | \ F^c_1, \dots, F^c_{q-1}, \mbox{ and } \tau_{q} \mbox{ occurs }) > b\ .
\end{equation}
Once we prove this, the assertion in the lemma will follow: We make the relevant
switch the first time a viable switching occurs, and  the probability
that this occurs before the packets reach the baths tends to $1$ as $L \to \infty$.

To prove (\ref{switchprob}), suppose the $q$th meeting
takes place at site $v$. Confusing $\hat{\bm Z}_n$ with $\bm Z_n$ 
momentarily, we assume for definiteness that packet $i$ is the first to arrive at site $v$, where it remains through time $\tau_q$ at which time packet $j$ is brought 
to site $v$ by particle $k'$. Under these assumptions, there are the following possibilities:

\smallskip \noindent 
{\it Case 1.} $\bm Z_{i,\tau_q}=v$. Since $\mathbb P(\bm Z_{j,\tau_q}
=v) = \frac12$, $\mathbb P(F_q) = \frac12$ in this case. 

\smallskip \noindent 
{\it Case 2.} $\bm Z_{i,\tau_q-1}=v$, and $\bm Z_{i,\tau_q}=k'$. In order for
a viable switching to occur, packets $i$ and $j$ must be the only packets carried
by particle $k'$ at the end of step $\tau_q$. There being a maximum of $N$
packets available to be picked up by particle $k'$ at this time, we have
\begin{equation} \label{exchbound}
\mathbb P(\bm Z_{i,\tau_q}=\bm Z_{j,\tau_q}=k' \mbox{ and }
 \bm Z_{\ell, \tau_q} \ne k' \ \forall \ell \ne i,j) \ge \frac{1}{N+1} \cdot \frac{2}{N(N-1)} \ .
 \end{equation}

\smallskip \noindent 
{\it Case 3.} $\bm Z_{i,\tau_q-1}=k$ for some particle $k$. Since $k \ne k'$,
for a viable switching to occur, packet $i$ must be the only packet picked up by
particle $k$ the last time it moved; that probability is $\ge \frac{1}{N(N+1)}$.
As packet $j$ must also be the only packet carried by particle $k'$
at the end of step $\tau_q$, we have $\mathbb P(F_q) \ge \left(\frac{1}{N(N+1)}\right)^2$.

\smallskip
We have thus a lower bound $b$ for $F_q$ that depends only on $N$.
\end{proof}

\smallskip
Notice that the argument above applies equally well to the setting of Theorem 4: 
even with their initial positions $O(L^\vartheta)$ apart,
we still have (\ref{eq:encounters}) beacuse of the gambler ruin's estimate and the argument above.
%the expected number of encounters between packets $i$ and $j$ tends to infinity 
%as $L \to \infty$. 

\smallskip
\begin{proof}[Proof of Theorems 2 and 4 for $d=1$] Let $v_1, v_2, \dots \in S$ and
$p_{L, v_1, \dots, v_N}$ be as defined at the beginning of this subsection. 
For each $N$, $p_{L,v_1, \dots, v_N}$ 
has a convergent subsequence as $L \to \infty$ by compactness. 
It follows that there exists $L_n \to \infty$ so that as $n\to \infty$,
$$
p_{L_n, v_1, \dots, v_N} \Rightarrow p_{\infty, v_1, \dots, v_N} \qquad
\mbox{ for every } N\ .
$$
Equivalently, if $X_1^{(L)}, X_2^{(L)}, \dots$ are $\{0, 1\}$-valued
random variables defined in such a way that when $\{0,1\}$ is identified 
with $\{\mathcal L, \mathcal R\}$,
the joint distribution of $(X_1^{(L)}, \dots, X_N^{(L)})$ is equal to
$p_{L,v_1, \dots, v_N}$ for each $N$, then there is a sequence of  
random variables $(X_1, X_2, \dots)$ to which $(X_1^{(L_n)}, X_2^{(L_n)}, \dots)$ converges.

By Lemma \ref{lemma:exch}, the sequence $(X_1, X_2, \dots)$ is exchangeable,
and de Finetti's Theorem applies (see Sect. 3.1). It follows from Sect. 6.1 that
if  $m$ is the probability in Theorem 5, then
\begin{eqnarray*}
x & = & \mathbb P(X_1=1) = \int_0^1 \theta \mathrm{d}m\ , \\
x^2 & = & \mathbb P(X_1 = X_2=1) = \int_0^1 \theta^2 \mathrm{d}m\ .
\end{eqnarray*}
From these two integrals, together with Jensen's Inequality, it follows that
$m=\delta_x$, the delta function at $x$. That in turn implies, by the characterization
of the measure $m$ in de Finetti's Theorem, that $X_1, \dots, X_N$ are
independent. Since the analysis above applies to any limit point of 
$p_{L, v_1, \dots, v_N}$
as $L \to \infty$, we conclude that $p_{L, v_1, \dots, v_N}$ converges to 
the product measure as claimed.
 \end{proof}

%%%%%%%%%%%%%%%%%%%%%%%%%%%%%%%%%%%%%
%%%%%%%%%%%%%%%%%%%%%%%%%%%%%%%%%%%%

\medskip
\section{LTE for particle numbers and energies}
\label{sec:dense}

This section is dedicated to the proof of Theorem \ref{thm:dense}.

We start with the following simple observation: The condition 
$M(L)/L^d \to \alpha \vol (\mathcal D)$ implies that for any $v \in \mathcal D_L$,
\begin{eqnarray*}
&& \mu^{(L)} \{\# \mbox{ particles at site } v = K \}  \\
&=& {{M(L)} \choose K} \left( \frac{1}{|\mathcal D_L|}\right)^K  \left( 1-\frac{1}{|\mathcal D_L|}\right)^{M(L)-K}
\rightarrow \ \ \frac{\alpha^K}{K!} e^{-\alpha} \quad \text{ as $L \to \infty$}.
\end{eqnarray*}
A similar computation involving finitely many sites proves statement (1) of Theorem \ref{thm:dense}.

We will give a proof of statement (2) of Theorem \ref{thm:dense} in the case
$S=\{ 0\}$. This assumption is not necessary, but it simplifies the notation
considerably, and the proof of the general case is entirely analogous. 

To fix notation, let $K \in \mathbb Z^+$ be the number of particles at site 
$\langle xL \rangle$, and fix  arbitrary nonnegative integers 
$n^*_0$, $\tilde n^*_{\mathcal [1]}, ..., \tilde n^*_{[K]}$, to be used as moments
of the site and particle energies.
For $L$ with
$M(L) \ge K$, let $(\mathcal T_1, ...,\mathcal T_K)$ be an ordered list of
distinct elements of $\{ 1,2,..., M(L)\}$. We let $Q(\mathcal T_1, ...,\mathcal T_K)$ 
denote the event that these are
exactly the particles at site $\langle xL \rangle$, and
define $\check{\underline{n}}_*=
\check{\underline{n}}_*(\mathcal T_1, ...,\mathcal T_K)$ by
\begin{equation}
\label{eq:n*choice2}
n_v = \left\{ \begin{array}{rl}
 n^*_0 &\mbox{ if $v = \langle xL \rangle$} \\
  0 &\mbox{ if $v \neq \langle xL \rangle$,}
       \end{array} \right. \quad
\hat n_v = 0 \ \ \forall  v, \quad 
\tilde n_j = \left\{ \begin{array}{rl}
 \tilde n^*_{[i]} &\mbox{ if $ j = \mathcal T_i$} \\
  0 &\mbox{ otherwise\ .}
       \end{array} \right.
\end{equation} 
%Note the difference between (\ref{eq:n*choice}) and (\ref{eq:n*choice2}): 
%in the proof of LTE for site energies only, $\tilde n_j=0$ for all $j$;
%now they can be positive.
%Let $Q(\mathcal T_1, ...,\mathcal T_K)$ denote the event that these are
%exactly the particles at site $\langle xL \rangle$.

\begin{lemma} \label{lemma:last} The following holds for all large enough $L$:
Let $(\mathcal T_1, ...,\mathcal T_K)$ be fixed. We define 
$\check{\underline{n}}_*$ 
as above, and let $N = n^*_0 + \sum_{j=1}^K \tilde n^*_{[j]}$. Then
\begin{equation} \label{F'}
 \int  F(\check{\underline{n}}_*,  \check{\underline{x}} ) 
1_{Q(\mathcal T_1, ...,\mathcal T_K)} \mathrm{d}\mu^{(L)} (\check{\underline{x}}, \bar X) \ \sim \ 
\mu^{(L)}(Q(\mathcal T_1, ...,\mathcal T_K)) \cdot [u(x)]^{N} \ .
\end{equation}
\end{lemma}

Observe that the quantities on both sides of (\ref{F'}) are 
independent of $(\mathcal T_1, ...,\mathcal T_K)$, 
as any two lists of $K$ particle names are clearly interchangeable,
so the ``$\sim$" can be interpreted without ambiguity as convergence as $L \to \infty$.
Let us denote the right side of (\ref{F'}) by $\mu^{(L)}(Q_K) \cdot [u(x)]^N$.

\begin{proof}[Proof of Theorem 3(2) assuming Lemma \ref{lemma:last} and $S=\{0\}$] 
 Let
$$A_K =  \{\# \mbox{ particles at site } \langle xL \rangle = K \}\ .
$$
We project $\mu^{(L)} | A_K$, the conditional 
measure of $\mu^{(L)}$ on $A_K$, to the site and particle energy coordinates 
on $\langle xL \rangle$. The resulting probability, $\nu^{(L)}_{x,K}$, will be
viewed as a measure on $\mathbb R^{K+1}$ with coordinates
$(\xi, \omega_1, ..., \omega_K)$, the $\xi$-coordinate corresponding to site
energy. 

Let $n^*_0$, $\tilde n^*_{\mathcal [1]}, ..., \tilde n^*_{[K]}$ be fixed. Using
the notation above, we have
\begin{eqnarray*}
&& \int \frac{\xi^{n^*_0}}{n^*_0 !} \ \prod_{j=1}^K \frac{\omega_j^{\tilde n^*_{ [j]}}}{n^*_{ [j]}!} \ \mathrm{d} \nu^{(L)}_{x,K}  \\
&= & \frac{1}{K !} \sum_{(\mathcal T_1, \dots, \mathcal T_K)}
\int  \frac{\xi_{\langle xL \rangle}^{n^*_0}}{n^*_0 !} \ 
\prod_{j=1}^K \frac{\eta_{\mathcal T_j}^{\tilde n^*_{ [j]}}}{n^*_{ [j]}!} \ 
1_{Q(\mathcal T_1, ..., \mathcal T_K)}\mathrm{d}(\mu^{(L)} ( \check{\underline{x}} , \bar{X})|A_K)\\
&= & \frac{1}{K !} \sum_{(\mathcal T_1, \dots, \mathcal T_K)}
\int  F(\check{\underline{n}}_*(\mathcal T_1, ..., \mathcal T_K),  \check{\underline{x}} ) 
1_{Q(\mathcal T_1, ..., \mathcal T_K)} \ \frac{1}{\mu^{(L)}(A_K)}
\mathrm{d}(\mu^{(L)}( \check{\underline{x}} , \bar{X}) ) \\
& \sim & {{M(L)} \choose K} \cdot \mu^{(L)}(Q_K) [u(x)]^N \cdot
\frac{1}{\mu^{(L)}(A_K)} \ = \  [u(x)]^{N}\ .
\end{eqnarray*}
The first two equalities are by definition. The convergence as $L \to \infty$
and the last
equality are from Lemma \ref{lemma:last} and the comments following that lemma.

As this holds for all $n^*_0$, $\tilde n^*_{\mathcal [1]}, ..., \tilde n^*_{[K]}$, 
we conclude, as in Sect. 4.3, that $\nu^{(L)}_{x,K}$ tends to a product of $K+1$ 
exponential distributions each with mean $u(x)$.
\end{proof}

\smallskip

\begin{proof}[Proof of Lemma \ref{lemma:last}]  Let
$(\mathcal T_1, ..., \mathcal T_K)$ be fixed.
We say a particle configuration $\sigma \in Q$ if $Q(\mathcal T_1, ..., \mathcal T_K)$ holds, and introduce the function
$$ F'(\check{\underline{n}}, \check{\underline{x}},\sigma ) = 
F(\check{\underline{n}},\check{\underline{x}} )
  1_{ \{ \sigma \in Q\} }\ .
  $$
The assertion in Lemma \ref{lemma:last} can then be rewritten as
\begin{equation}
\label{eq:dense2}
\int F'(\check{\underline{n}}_*, \check{\underline{x}},\bar X ) \mathrm{d}\mu^{(L)} (\check{\underline{x}}, \bar X) \sim 
\mu^{(L)}\{\bar X \in Q\} \cdot [u(x)]^{N}.
\end{equation}
As in Section \ref{sec:duality}, we  approximate the left side of 
(\ref{eq:dense2}) by 
$\mathbb E(F'(\check{\underline{n}}_*,\check{\underline{x}}_t, \bar X_t))$,
and decompose into sample paths $\bm \sigma=(\sigma_0, \dots, \sigma_m)$ 
of particle movements on $[0,t]$:
\begin{eqnarray}
&&\int \mathbb E(F'(\check{\underline{n}}_*,\check{\underline{x}}_t, \bar X_t )| 
\check{\underline{x}}_0 = \check{\underline{x}}_*, \bm \sigma) \mathbb
P (\mathrm{d}\bm \sigma)\label{denseformula}\\
&=& \sum_m
\sum_{\substack{\bm \sigma = (\sigma_0, ..., \sigma_m)\\ \sigma_m \in Q}} 
\mathbb E(F(\check{\underline{n}}_*,\check{\underline{x}}_t)| 
\check{\underline{x}}_0 = \check{\underline{x}}_*, \bm \sigma) \mathbb P (\bm \sigma)\nonumber\\
&=& \sum_m
\sum_{\substack{\bm \sigma = (\sigma_0, ..., \sigma_m)\\ \sigma_m \in Q}}
\mathbb E( F(\underline{\check{n}}_t, \underline{\check{x}}_*) | \underline{\check{n}}_0 = \underline{\check{n}}_*, \bm \sigma^{-1})
 \mathbb P (\bm \sigma^{-1})\ . \label{denseformula2}
\end{eqnarray}
The first equality above is the definition of $F'$ and the second follows from
Lemma \ref{lemmaduality2}. Informally, for an $\bm X_t$-trajectory 
that ends in a state with a set of particles at a certain site, its ``dual trajectory",
which is obtained by reversing the paths of the particles, should {\it start} with
the same set of particles at the same site.

Taking the limit $t \to \infty$ in lines (\ref{denseformula}) and 
 (\ref{denseformula2}) we obtain
\begin{equation} \label{E'} 
\int F'(\check{\underline{n}}_*, \check{\underline{x}},\bar X ) \mathrm{d}\mu^{(L)} (\check{\underline{x}}, \bar X)=
\mathbb E \left( \prod_{i=1}^N T \left( 
\frac{\bm Z_{i,\infty}^{(L)}}{L}
\right) 1_{\{\bm Z_0 \in Q \}}\right).
\end{equation}
To handle the right side of (\ref{E'}), we alter the definition of $\bm Z_k$ in 
Sect. 4.3 slightly by restricting $\bar Y$ 
in $\bm Z_0$ to $Q$, everything else unchanged. Let $U_Q$ denote this 
new probability distribution of $\bar Y$, and define
$\mathbb E'(\cdot) = \int \mathbb E_{(\check{\underline{n}}_*,\bar Y)}(\cdot) 
\ U_Q(\mathrm{d}\bar Y)$ (cf (\ref{exp-exp}) in Sect. 4.3), so that the right side of (\ref{E'})
is equal to
$$
\mu^{(L)}(\{\bar Y \in Q\} )\cdot \mathbb E' \left( \prod_{i=1}^N T \left( 
\frac{\bm Z_{i,\infty}^{(L)}}{L}\right)\right)\ .
$$
To obtain (\ref{eq:dense2}), it remains to check that Proposition 4 holds 
with $\mathbb E$ replaced
by $\mathbb E'$:  For $d \ge 2$, this is not an issue, since properties of $\bar Y$ do 
not appear in the proof. For 1d, one needs to check that the switching
arguments in Sect. 6.2 are not affected by the restriction of $\bar Y$ to $Q$,
and that is true as well. 
\end{proof}

\medskip

\section*{Appendix: Proof of Proposition \ref{lem12}}

Our proof uses the continuous mapping theorem, which can be stated as follows
 (see e.g. Theorem 5.1
 in \cite{B68}). As before, weak convergence is denoted by ``$\Rightarrow$", and for
a mapping $f$ and measure $\mu$, $f_*\mu$ is the measure given by
 $f_*\mu(A)=\mu(f^{-1}(A))$.

\smallskip
\begin{theorem}[{\rm Continuous mapping theorem}]
\label{thm:ctsmap}
Let $X$ and $X'$ be separable metric spaces, and let $P(X)$ denote the
set of Borel probability measures on $X$. We consider a Borel measurable mapping
 $f:X \rightarrow X'$ with discontinuity set $D_f$, and let $\mu_n, \mu \in P(X)$
be such that (i) $\mu_n \Rightarrow \mu$ as $n \to \infty$ and (ii) $\mu(D_f) = 0$.
Then $f_*{\mu_n} \Rightarrow f_*\mu$ as $n \to \infty$.
\end{theorem}

\smallskip

To fix some notation, for $\mathfrak T \in \mathbb R^+$ we let
$$ X = C([0, \mathfrak T], \mathbb R^d) $$
be the set of continuous maps from $[0, \mathfrak T]$ to
$\mathbb R^d$ endowed
with the sup norm, making it a separable metric space. 
For $a \in \mathbb R^d$, we let $B^a$ be the standard 
Brownian motion starting from $a \in \mathbb R^d$ up to time $\mathfrak T$,
and with a slight abuse of notation, we use $B^a$ to denote also the 
corresponding measure on $X$. Let $W^a_n$ be 
the rescaled SSRW up to time $\lfloor n\mathfrak T \rfloor$ starting 
from $a \in \mathbb R^d$, 
i.e. if $\hat S_k$ is a $d$-dimensional SSRW with $\hat S_0=\langle a \sqrt n \rangle$, then
$W^a_n(k/n) = \frac{\hat S_k}{\sqrt n}$ for $k = 0,1,..., \lfloor n\mathfrak T \rfloor$,
and $W^a_n(t)$ is obtained by 
interpolating linearly between $ t = k/n$ and $t = (k+1)/n$. 
As with $B^a$, we use $W^a_n$ to denote also the corresponding measure on $X$. 
By the invariance principle, 
\begin{equation} \label{invpr}
W^a_n \Rightarrow B^a \qquad \mbox{ as } n \to \infty
\end{equation}
on any interval $[0, \mathfrak T]$. The convergence in (\ref{invpr}) differs from
that asserted in  
Proposition \ref{lem12} in that the latter is for paths that terminate
not at a fixed time but upon reaching $\partial \mathcal D$.

\begin{proof}[Proof of Proposition \ref{lem12}]
Let $x \in \mathcal D$ and $\varepsilon>0$ be given. We write
$$ \mathcal D^1 = \{ y \in \mathbb R^d | \exists z \in \mathcal D, |y-z|<1\},$$
and let $\mathfrak T = \mathfrak T (\varepsilon)$ be such that
a Brownian motion starting from $x$ reaches $\partial \mathcal D^1$ before time $\mathfrak T $ with probability at least $1-\varepsilon/(2\| T\|_{\infty})$.
For $\omega \in X=C([0, \mathfrak T], \mathbb R^d)$, we define
$\tau(\omega) = \min \{ t \in [0, \mathfrak T]: \omega(t) \in \partial \mathcal D \}$
if such a $t$ exists, $= \infty$ if it does not. Then we define $f: X \to \mathbb R$ by
$$
f( \omega) \ = \  \left\{ \begin{array}{rl}
 T(\omega(\tau)) &\mbox{ if $\tau\leq \mathfrak T$} \\
   \max_{y \in \partial \mathcal D} T(y) &\mbox{ if } \tau = \infty\ .
       \end{array} \right.
$$

\begin{lemma} \label{discontinuity} $B^x(D_f) = 0$
\end{lemma}

We first finish the proof assuming the result of this lemma.
Via a rigid translation, we may assume $x =0$ 
(so that $xL \in \mathbb Z^d$ for all $L$). Then with $S_0=xL$, $S_n/L$
in the proposition is $W^0_{L^2} (n/L^2) $ in the notation above.
Observe that we are now in the setting of the continuous mapping theorem: 
$W^0_{L^2}  \Rightarrow B^0$ as $L \to \infty$ is condition (i) in 
Theorem \ref{thm:ctsmap}, and the assertion in Lemma \ref{discontinuity} is condition (ii). Thus the theorem applies, and its conclusion together with our choice of $\mathfrak T$ gives exactly (\ref{approx}) in the case $S_0=xL$. 

To prove the full statement of Proposition \ref{lem12}, we observe that
if the result was false, there would be  a sequence $x_k \in \mathbb R^d$ with $x_k \to x$ 
and a sequence $L_k \to \infty$ such
that if $S_n^{(k)}$ is the SSRW on $\mathcal D_{L_k}$ with $S^{(k)}_0 = x_k L_k$,
then
$$
 \left| \mathbb E \left(T \left( \frac{S^{(k)}_{\tau}}{L_k} \right) \right) -u(x)\right| 
 > \varepsilon
$$
where $\tau$ is the smallest $n$ such that $S^{(k)}_n \in \mathcal B_{L_k}$.
Such a scenario cannot occur: Since $W^{x_k}_{L_k^2} = (x_k-x) + W^x_{L_k^2}$,
it follows from (\ref{invpr}) that $W^{x_k}_{L_k^2} \Rightarrow B^x$ on 
 $[0, \mathfrak T]$, and the argument 
in the last paragraph with $W^{x_k}_{L_k^2}$ in the place of $W^0_{L^2}$ 
gives the opposite inequality. 
\end{proof}

\smallskip
To complete the proof, it remains to show that the discontinuity set of $f$ has
zero Wiener measure.

\begin{proof}[Proof of Lemma \ref{discontinuity}] First, we identify the discontinuity
set $D_f$. If $\tau(\omega)
= \infty$, then the trajectory of $\omega$ up to time $\mathfrak T$ is bounded away
from $\partial \mathcal D$, hence $f$ is continuous at $\omega$. If 
$\tau(\omega) < \mathfrak T$, then $\liminf_{\omega' \to \omega}
\tau(\omega') \ge \tau(\omega)$ for the same reason, but 
the corresponding $\limsup$ can be strictly greater
than $\tau (\omega)$  if  the trajectory of 
$\omega$ does not cross to the other side of $\partial \mathcal D$ 
immediately following $\tau(\omega)$. More precisely, we have deduced that
$D_f = \{\tau  = \mathfrak T\} \cup E$ 
where
$$ E = 
\{\omega : \tau (\omega) < \mathfrak T \mbox{ and } \exists \eta=\eta(\omega) >0 \mbox{ s.t. } \omega((\tau, \tau+\eta)) \subset \bar{\mathcal D}\} 
$$
where $\bar{\mathcal D}$ is the closure of $D$. 

Clearly, $\{ \tau = \mathfrak T\}$ has measure $0$, so it suffices to show $B^x(E)=0$.

Since harmonic measure is absolutely continuous, the set of $\omega$
for which $\omega (\tau)$ lies at a point at which $\partial \mathcal D$ is not
$C^2$ differentiable has measure zero. Let $\omega$ be outside of this
measure zero set, and fix an orthonormal basis $\{e_1, ..., e_d\}$ 
of $\mathbb R^d$ such that $e_1$ is the outward normal to $\partial \mathcal D$
at $\omega(\tau)$. Then there exist $K>0$ and a neighborhood $U$ of $\omega(\tau)$ in $\mathbb R^d$ such that
\begin{equation} \label{paraboloid}
 y \in  \bar{\mathcal D} \quad \text{ implies } \quad y_1 <  K \|
 (y_2, ... ,  y_d)\|^2\ .
\end{equation}
%Consider $\omega$ with $\omega(\tau) \in U$, and observe that 
Recall that by the strong Markov property, $B^x$ starting from the stopping time $\tau $ is a Brownian motion. In particular, by projecting this Brownian motion, which we call
$\hat B(t)$, to the line parallel to $e_1$ and to the hyperplane spanned by $e_2, ..., e_d$, we obtain two independent Brownian motions, $\hat B_1(t)$ 
and $\hat B_{d-1}(t)$. Since $\hat B_1(t)/\sqrt t$ and
$\hat B_{d-1}(t)/\sqrt t$ have standard normal distributions, 
$$ \mathbb P(|\hat B_1(t)| < t^{2/3}) \rightarrow 0 \quad \text{ and } \quad 
\mathbb P(\|\hat B_{d-1}(t)\| > t^{1/3}/\sqrt{K}) \rightarrow 0 
$$
as $t \to 0$. Choosing $t_n \downarrow 0$ so that
$$ \sum_n \mathbb P(|\hat B_1(t_n)| < t_n^{2/3}) \ ,
\ \  \sum_n
 \mathbb P(\|\hat B_{d-1}(t_n)\| > t_n^{1/3}/\sqrt{K}) \ < \ \infty,$$
it follows from the Borel-Cantelli lemma that
$|\hat B_1(t_n)| > t_n^{2/3}$ and $\|\hat B_{d-1}(t_n)\| < t_n^{1/3}/\sqrt{K}$ hold for
all but finitely many $n$. For $t_n$ for which these inequalities hold, we are
guaranteed that $\hat B(t_n) \not \in \bar{\mathcal D}$ if $\hat B_1(t_n)>0$
and $\hat B(t_n) \in U$.

Now it is a well known fact that on the time interval $[0, \eta]$ for every $\eta>0$,
a 1D Brownian motion starting from $0$ makes infinitely many excursions from $0$,
and each excursion is positive with probability $1/2$ independently of other excursions. Applying this fact to $\hat B_1(t)$, and assuming (as we may) that
each $t_n$ lies in a different excursion, it follows that with probability $1$,
$\hat B_1(t_n)>0$ for infinitely many $n$. 
Since $\mathbb P(\hat B(t) \in U, t \in [0, \eta]) \to 1$ as $\eta \to 0$,
we have proved that 
$ B^x (E) = 0$.

%$$\mathbb P(\omega \in E | 
%\tau(\omega) < \mathfrak T \mbox{ and } \omega(\tau) \in U)=0\ .$$ 
\end{proof}

%%%%%%%%%%%%%%%%%%%%%%%%%%%%%%%%%%%%%%%%%%%%%%%%%%%
%%%%%%%%%%%%%%%%%%%%%%%%%%%%%%%%%%%%%%%%%%%%%%%%%%%%%


\medskip

\begin{thebibliography}{99}


\bibitem[B68]{B68} Billingsley, P., Convergence of probability measures, {\it Wiley} 1968.

\bibitem[BO05]{BO05} Bernardin, C., Olla, S. Fourier's law
for a microscopic model of heat conduction
{\it Journal of Stat. Phys.} {\bf 121} 271--289 (2005)

\bibitem[BLY10]{BLY10}
P.~Balint, K.K. Lin, and L.S. Young.
\newblock Ergodicity and energy distributions for some boundary driven
  integrable hamiltonian chains.
\newblock {\em Communications in Mathematical Physics}, 294(1):199--228, 2010.

\bibitem[BCS12]{BCS12}
Alexei Borodin, Ivan Corwin, and Tomohiro Sasamoto.
\newblock From duality to determinants for q-tasep and asep.
\newblock {\em arXiv preprint arXiv:1207.5035}, 2012.

\bibitem[BK07]{BK07}
J.~Bricmont and A.~Kupiainen.
\newblock Towards a derivation of fourierÕs law for coupled anharmonic
  oscillators.
\newblock {\em Communications in mathematical physics}, 274(3):555--626, 2007.

\bibitem[BLL04]{BLL04}
Bonetto, F., Lebowitz, J.L., Lukkarinen, J. , Fourier's law for a
harmonic crystal with self-consistent stochastic reservoirs. {\it
  Journal of statistical physics} {\bf 116}1-4 (2004): 783-813.

\bibitem[BLRB00]{BLRB00}
Bonetto, F., Lebowitz, J. L., Rey-Bellet, L., 
Fourier's law: a challenge to theorists, Proceedings ICMP-2000, Imp. Coll. Press, London, pp. 128--150.



\bibitem[CGGR13]{CGGR13} Carinci, G., Giardin\`a, C., Giberti, C., Redig, F.,
Duality for stochastic models of transport, {\it Journal of
Statistical Physics} (2013).
%Best source for reviewing duality in various models. Topics 1(a),3,4,5. 

\bibitem[CH49]{CH49} Chung, K.L., Hunt, C.A.,
On the zeros of $\sum_1^n \pm 1$ {\it Annals of Mathematics} {\bf 50} 385--400 (1949)

%\bibitem[C07]{C07} Carmona, P., Existence and uniqueness of an invariant measure for a chain of
%  oscillators in contact with two heat baths.
%{\it Stochastic processes and their applications} {\bf 117} 8 :1076--1092, (2007).

\bibitem[CE14]{CE14}
No{\'e} Cuneo and J-P Eckmann.
\newblock Controlling general polynomial networks.
\newblock {\em Communications in Mathematical Physics}, 328(3):1255--1274,
  2014.



\bibitem[D31]{D31} de Finetti, B.: 
Funzione caratteristica di un fenomeno aleatorio. 
{\it Atti della R. Academia Nazionale dei Lincei, Serie 6. Memorie, Classe di Scienze Fisiche, Mathematice e Naturale} {\bf 4} 251 - 299. (1931)

\bibitem[DL11]{DL11} Dolgopyat, D., Liverani, C., 
Energy transfer in a fast-slow Hamiltonian system 
{\it Communications in Mathematical Physics} {\bf 308} 1,
201--225 (2011).

\bibitem[DN14]{DN14} Dolgopyat, D., N\'{a}ndori, P.,
Non equilibrium density profiles in Lorentz tubes with thermostated boundaries {\it preprint} (2014)

\bibitem[DD99]{DD99}
 A.~Dhar and D.~Dhar, ``Absence of local thermal
  equilibrium in two models of heat conduction,''  {\it
    Phys. Rev. Lett.} {\bf 82} (1999) 480--483

\bibitem[D07]{D07} 
B. Derrida, ``Non-equilibrium steady states:
  fluctuations and large deviations of the density and of the current,''
  {\em J. Stat. Mech.}  (2007) P07023

\bibitem[DLS02]{DLS02}
B.~Derrida, J.~L.~Lebowitz, E.~R.~Speer, ``Large
  deviation of the density profile in the steady state of the open
  symmetric simple exclusion process,'' {\em J. Statist. Phys.} {\bf
  107} (2002) {\em pp.}~599--634

 
\bibitem[ET60]{ET60} Erd\H{o}s, Taylor: Some problems concerning the structure of random walk paths
{\it Acta Math. Acad. Sci. Hungar.} {\bf 11} 137 - 162 (1960)

\bibitem[EPR99]{EPR99}
J-P Eckmann, C-A Pillet, and L~Rey-Bellet, \emph{Non-equilibrium statistical
  mechanics of anharmonic chains coupled to two heat baths at different
  temperatures}, Communications in Mathematical Physics \textbf{201} (1999),
  no.~3, 657--697.

\bibitem[EY06]{EY06}
J.P. Eckmann and L.S. Young.
\newblock Nonequilibrium energy profiles for a class of 1-d models.
\newblock {\em Communications in mathematical physics}, 262(1):237--267, 2006.


\bibitem[ELS90]{ELS90} G. Eyink, J. L. Lebowitz and H. Spohn,
 Hydrodynamics of Stationary Nonequilibrium States
for Some Lattice Gas Models, 
\newblock {\em Commun. Math. Phys.}, Vol. 132, 252--283 (1990).

\bibitem[ELS91]{ELS91} G. Eyink, J. L. Lebowitz and H. Spohn, Lattice Gas Models in contact with Stochastic Reservoirs:
Local Equilibrium and Relaxation to the Steady State, 
\newblock {\em Commun. Math. Phys.}, Vol. 140, 119--131 (1991).

\bibitem[FLM11]{FLM11} J. Farfan, C. Landim, M. Mourragui,
Hydrostatics and dynamical large deviations of boundary 
driven gradient symmetric exclusion processes. 
\newblock {\em Stochastic Process. Appl.} Vol. 121, 725--758 (2011).

\bibitem[GG08]{GG08}
P.~Gaspard and T.~Gilbert.
\newblock Heat conduction and fourier's law in a class of many particle
  dispersing billiards.
\newblock {\em New Journal of Physics}, 10(10):103004, 2008.

\bibitem[GKR07]{GKR07}
Cristian Giardina, Jorge Kurchan, and Frank Redig.
\newblock Duality and exact correlations for a model of heat conduction.
\newblock {\em Journal of mathematical physics}, 48(3):033301, 2007.

\bibitem[GM62]{GM62}
S.~R.~de Groot and P.~Mazur, {\em Non-equilibrium
  Thermodynamics}, North-Holland (1962)



\bibitem[JK14]{JK14} Jansen, S., Kurt, N.: On the notion(s) of duality 
for Markov processes, {\it Probability surveys} 11, (2014)


\bibitem[KMP82]{KMP82} C. Kipnis, C. Marchioro, E. Presutti: Heat flow in an exactly solvable model
{\it Journal of Stat. Phys.} {\bf 27} 65-74 (1982)

\bibitem[KLS84]{KLS84} Katz, S., Lebowitz, J., Spohn, H.,
Nonequilibrium steady states of stochastic lattice gas models of fast ionic conductors
{\it Journal of Stat. Phys.} {\bf 34} 3/4 497--537 (1984)

\bibitem[KY13]{KY13}
Khanin Konstantin and Yarmola Tatiana.
\newblock Ergodic properties of random billiards driven by thermostats.
\newblock {\em Communications in Mathematical Physics}, 320(1):121--147, 2013.

\bibitem[KL99]{KL99}
C. Kipnis and C. Landim, {\em Scaling Limits of
  Interacting Particle Systems}, Berlin: Springer-Verlag (1999)



\bibitem[L91]{L91} Lawler, G.: Intersections of random walks. {\it Birkh\"auser}, Boston (1991)

\bibitem[L85]{L85} Liggett, T.M. Interacting particle systems
{\it Springer, Berlin} 1985.

\bibitem[LY10]{LY10} Lin, K., Young, L.-S.,
Nonequilibrium steady states for certain hamiltonian models.
{\it Journal of Statistical Physics} {\bf 139} 4, 630--657 (2010).

\bibitem[LO12]{LO12} Liverani, C., Olla, S., 
Toward the Fourier law for a weakly interacting anharmonic crystal,
{\it Journal of the American Mathematical Society} {\bf 25} 555--583 (2012).

\bibitem[LLP03]{LLP03}
S.~Lepri, R.~Livi, and A.~Politi.
\newblock Thermal conduction in classical low-dimensional lattices.
\newblock {\em Physics Reports}, 377(1):1--80, 2003.

\bibitem[LLM03]{LLM03}
H.~Larralde, F.~Leyvraz, and C.~Mejia-Monasterio.
\newblock Transport properties of a modified lorentz gas.
\newblock {\em Journal of statistical physics}, 113(1):197--231, 2003.

\bibitem[LY14]{LY14}
Yao Li and Lai-Sang Young.
\newblock Nonequilibrium steady states for a class of particle systems.
\newblock {\em Nonlinearity}, 27(3):607, 2014.



\bibitem[OR13]{OR13} Opoku, A., Redig, F., Coupling independent walkers and the inclusion processes
{\it preprint} (2013) http://arxiv.org/abs/1311.1620

\bibitem[P75]{P75} Petrov: Sums of independent random variables
{\it Akademie-Verlag, Berlin} (1975)



\bibitem[RY07]{RY07} Ravishankar, Young: Local Thermodynamic Equilibrium for some Models of Hamiltonian Origin
{\it Journal of Stat. Phys.} {\bf 128} 3 (2007)

\bibitem[PS83]{PS83}
Errico Presutti and Herbert Spohn.
\newblock Hydrodynamics of the voter model.
\newblock {\em The Annals of Probability}, pages 867--875, 1983.



\bibitem[R11]{R11} Ruelle, D., A mechanical model for 
Fourier's law of heat conduction
{\it Communications in Mathematical Physics} {\bf 311}
3, 755--768 (2012).

\bibitem[RT02]{RT02}
Luc Rey-Bellet and Lawrence~E Thomas.
\newblock Exponential convergence to non-equilibrium stationary states in
  classical statistical mechanics.
\newblock {\em Communications in mathematical physics}, 225(2):305--329, 2002.

\bibitem[RY12]{RY12}
B.~Ryals and L.S. Young.
\newblock Nonequilibrium steady states of some simple 1-d mechanical chains.
\newblock {\em Journal of Statistical Physics}, pages 1--15, 2012.

\bibitem[RLL67]{RLL67}
Z. Rieder, J. Lebowitz and E. Lieb,
``Properties of
  a harmonic crystal in a stationary nonequilibrium state,'' {\it
    J. Math. Phys.} {\bf 8} (1967)



\bibitem[S70]{S70} Spitzer, F., Interaction of Markov processes, 
Adv. Math. 5, 246--290 (1970)

%\bibitem[S80]{Sp80} Spohn, H.,
%Kinetic equations from Hamiltonian dynamics: Markovian limits.
%{\it Rev. Modern Phys.} {\bf 52} (1980) 569--615.
%Survey of a huge number of models and results. Topics 1,2.

\bibitem[S91]{S91} Spohn, H., Large scale dynamics of interacting particles, Springer, Berlin, New York, 1991.

\bibitem[Z83]{Z83} Zessin, H., The Method of Moments for Random Measures,
{\it Z. Wahrscheinlichkeitstheorie verw. Gebiete} {\bf 62} 395--409 (1983). 













\end{thebibliography}
\end{document}